\documentclass[journal]{IEEEtran}
\IEEEoverridecommandlockouts

\makeatletter
\def\markboth#1#2{\def\leftmark{\@IEEEcompsoconly{\sffamily}\MakeUppercase{\protect#1}}%
\def\rightmark{\@IEEEcompsoconly{\sffamily}\MakeUppercase{\protect#2}}}
\makeatother

\usepackage[utf8x]{inputenc}
\usepackage[english]{babel}
\usepackage{centernot}
\usepackage{amssymb,amstext,amsmath}
\usepackage{cases}
\usepackage{amsthm}
\usepackage{algorithm}
\usepackage{algorithmic}
\usepackage{url}
  \usepackage{enumerate}
\usepackage{cite}
\usepackage{etoolbox}
\usepackage{capt-of}

\ifCLASSINFOpdf
  \usepackage[pdftex]{graphicx}
  \usepackage{subfigure}
  \usepackage{epstopdf}
\else
  \usepackage[dvipdf]{graphicx}
\fi

\newtoggle{conference}
\toggletrue{conference} 
\makeatletter
\newcommand\remembertext[2]{
  \immediate\write\@auxout{\unexpanded{\global\long\@namedef{mytext@#1}{#2}}}%
  #2%
}
\newcommand\recalltext[1]{%
  \ifcsname mytext@#1\endcsname
    \@nameuse{mytext@#1}%
  \else
    ``??''
  \fi
}
\makeatother

\newcommand{\Hb}{\mathbf{H}}

\newcommand{\I}{\mathbf{I}}

\newcommand{\D}{\mathbf{D}}

\newcommand{\h}{\mathbf{h}}
\newcommand{\x}{\mathbf{x}}

\newcommand{\uu}{\mathbf{u}}
\newcommand{\y}{\mathbf{y}}

\newcommand{\z}{\mathbf{z}}

\newcommand{\tr}{\textnormal{tr}}

\newcommand{\Ex}[2]{{\textnormal{E}_{#1}\left[#2\right]}}

\newtheorem{theorem}{Theorem}
\newtheorem{remark}{Remark}
\newtheorem{corollary}{Corollary}
\theoremstyle{definition}
\newtheorem{definition}{Definition}

\title{Capacity Scaling of Cellular Networks:\\ Impact of Bandwidth, Infrastructure Density\\ and Number of Antennas}

\author{
    Felipe G\'omez-Cuba,~\IEEEmembership{Member,~IEEE},
    Elza Erkip,~\IEEEmembership{Fellow,~IEEE},\\
    Sundeep Rangan,~\IEEEmembership{Fellow,~IEEE},
    Francisco J. Gonz\'alez-Casta\~no
    \thanks{Parts of this work have appeared in conference versions \cite{Gomez-Cuba2014isit,gomez2016capacity}.}
    \thanks{The work of F. G\'omez-Cuba was supported by FPU12/01319MINECO, Spain; and EC H2020-MSCA-IF-2015 704837. The works of S. Rangan and E. Erkip were supported by NSF grants 1116589, 1547332 and 1302336; and in part by the industry affiliates of NYU WIRELESS. The work of F.J. Gonz\'alez-Casta\~no was supported by TEC2016-76465-C2-2-R, MINECO; and GRC2014/046, Xunta de Galicia, Spain. This project has received funding from the European Union's Horizon 2020 research and innovation programme under the Marie Sk\l{}odowska-Curie grant agreement No 704837.}
    \thanks{
        F. G\'omez-Cuba was with AtlantTIC, University of Vigo, EE Telecomunicaci\'on, 36310 Vigo, Spain and is now with Dipartimento di Ingegneria dell'Informazione, University of Padova, Via Gradenigo 6/b, 35131 - Padova Italy, and with Department of Electrical Engineering, Stanford University, 350 Serra Mall, 94305 CA USA (e-mail: {\tt gmzcuba@stanford.edu}),         
        E. Erkip  and S. Rangan are with Department of Electrical and Computer Engineering, New York University Tandon
        School of Engineering, Brooklyn, NY 11201 USA
        (e-mail: {\tt \{srangan,elza\}@nyu.edu}),
        F. J. Gonz\'alez-Casta\~no is with AtlantTIC, University of Vigo,
        EE Telecomunicaci\'on, 36310 Vigo, Spain (e-mail:
        {\tt javier@gti.uvigo.es})
    }
}

\begin{document}
\maketitle
\markboth{DRAFT}{DRAFT}
\begin{abstract}
The availability of very wide spectrum in millimeter wave bands combined with large antenna arrays and ultra dense networks  raises two basic questions: What is the true value of overly abundant degrees of freedom and how can networks be designed to fully exploit them? This paper determines the capacity scaling of large cellular networks as a function of bandwidth, area, number of antennas and base station density. It is found that the network capacity has a fundamental bandwidth scaling limit, beyond which the network becomes power-limited. An infrastructure multi-hop protocol achieves the optimal network capacity scaling for all network parameters. In contrast, current protocols that use only single-hop direct transmissions can not achieve the capacity scaling in wideband regimes except in the special case when the density of base stations is taken to impractical extremes.  This finding suggests that multi-hop communication will be important to fully realize the potential of next-generation cellular networks. Dedicated relays, if sufficiently dense, can also perform this task, relieving user nodes from the battery drain of cooperation. On the other hand, more sophisticated strategies such as hierarchical cooperation, that are essential for achieving capacity scaling in \emph{ad hoc} networks, are unnecessary in the cellular context. \end{abstract}


\section{Introduction}
\label{sec:introduction}


To meet the tremendous growth in demand for cellular wireless data, three new design approaches are widely-considered for the evolution of next-generation systems~\cite{BocHLMP:14}:
\begin{itemize}
\item Vast spectrum available at very high frequencies, esp.\ the millimeter wave \cite{Pi2011,PietBRPC:12,Rappaport2013,RanRapEr:14,Rappaport2014-mmwbook};
\item Massive Multiple-input Multiple-output (MIMO)  for increased spatial multiplexing \cite{hoydis2013massive,larsson2014massive};
\item Ultra dense deployments of small pico- and femtocells \cite{chandrasekhar2008femtocell,andrews2012femtocells}.
\end{itemize}

Together, these technologies offer the potential of orders of magnitude increases in capacity, and, if successful, may fundamentally change the basic constraints that dictate network design today. This possibility leads to two basic questions: What is the fundamental capacity offered by these technologies and how can networks be best designed to fully leverage their potential?

From an information theoretic perspective, millimeter wave transmissions, massive MIMO and ultra-dense deployments are all, in essence,
various ways to increase the fundamental degrees of freedom of the network which are controlled by bandwidth, number of antennas and infrastructure density respectively. This paper attempts to characterize the capacity scaling of cellular networks as a function of the scaling of these dimensions. Our analysis follows along the lines of the classic result of Gupta and Kumar \cite{kumar2000capacity}, but applied to cellular networks rather than \emph{ad hoc} networks with or without infrastructure.
Specifically, we consider a large cellular network with $n$ mobile nodes, where the key parameters such as bandwidth, number of antennas, area, and base stations (BSs), all scale as functions of $n$. In addition, traffic in this network travels between each one of the BSs and the nodes in its cell, in separate uplink and downlink phases.

Our main results determine the capacity scaling by finding identically-scaling lower and upper bounds on the throughput. The upper bound is a series of cut-set bounds in which one transmitter is cut from the rest of the network, and all the nodes and BSs in the other side of the cut cooperate perfectly, forming a virtual point-to-point MIMO system where all devices contribute to receive power and all interference is perfectly canceled. The capacity scaling achieving lower bound is found by considering a simple infrastructure multi-hop (IMH) protocol where transmissions are relayed to/from the closest BSs via mobile nodes within the same cell. 

We also study the capacity scaling of two additional protocols: The first one is the infrastructure single-hop (ISH) protocol, where transmissions are sent directly between the BS and each node within its cell, and which is the dominant paradigm in current cellular networks. The second one is the infrastructure relay multi-hop (IRH), modeled after existing two-layer network architectures, where IMH is used for wireless backauling of additional access points called relay nodes (RN), while user nodes only communicate directly with a single nearest access point using ISH to prevent multi-hop implementation difficulties due to mobility, reticence to cooperation, and backwards compatibility.

Our analysis yields several important and in some cases surprising findings:
\begin{itemize}
\item \emph{Bandwidth scaling limit:}
There is a ``critical bandwidth scaling'' that defines a maximum useful bandwidth for the whole network. Below the critical point, the capacity scales with the bandwidth, whereas if bandwidth grows faster than its critical limit the capacity becomes power-limited and additional bandwidth growth no longer improves the capacity scaling. Power and bandwidth limited regimes are well-understood for point-to-point channels, and our results provide a generalization to cellular networks.

\item \emph{Benefits of increased cell density:}
The network capacity always grows with the BS density, whereas the benefits of increased bandwidth or number of BS antennas have a limit. This is valid as long as nodes are sufficiently separated to experience far-field propagation.

\item \emph{Interference alignment is not necessary:}
Our upper bound implicitly avoids inter-cell interference, whereas our lower bound IMH simply treats interference as noise. Since both have the same scaling, we can conclude that interference-alignment schemes, despite providing significant gains in a non-asymptotic regime \cite{cadambe2008interference}, do not alter the capacity scaling significantly. On the other hand, our analysis does not discard that BS cooperation, achieved for example by a wired backhaul, could improve the capacity scaling over the non-cooperative BS model either with or without interference canceling. We leave this analysis for future work.

\item \emph{Multi-hop is optimal, and outperforms single-hop communication:}
The IMH protocol achieves the optimal capacity scaling in all regimes. ISH is optimal at small bandwidth scaling but performs strictly worse than IMH in regimes with wide bandwidths or large numbers of antennas. The reason is that ISH employs longer transmission distances and becomes power-limited earlier than IMH as bandwidth scaling is increased. This suggests that, even though in today's networks capacity is bandwidth-limited and direct transmissions between the mobile nodes and the BS are efficient, in future networks with much larger bandwidths, multi-hop communication may be necessary to fully achieve the network capacity.

\item \emph{Hierarchical Cooperation is not necessary in cellular systems:}
Optimality of IMH implies that Hierarchical Cooperation (HC) cannot improve the rate scaling achieved with IMH, as opposed to dense ad-hoc networks, where multi-hop is optimal only in some regimes and HC is necessary to achieve the optimal throughput scaling otherwise \cite{ozgur2009information}. 

\item \emph{Wireless backhauling may be optimal, but RN density is critical:}
IRH performance depends on the RN density. In the best case scenario with RNs as dense as the user nodes, IRH rate scales as IMH and can be regarded as a practical strategy to achieve capacity scaling while avoiding mobility issues. But if the RN density is lower, the performance of IRH is suboptimal in the power-limited regimes with high bandwidth scaling, and may not offer any gains over ISH in the bandwidth-limited regime with low bandwidth scaling.

\item \emph{Applicability to fading and non-coherent communications:}
\remembertext{Lozclarification}{The main results in this paper are obtained under a deterministic path loss channel model with full rank and additive Gaussian noise. However, very similar network scaling laws can be readily argued for the case of the ergodic capacity in frequency selective fading channels with channel state information (CSI) at all network nodes. For the case of non-coherent fading channels --without CSI-- there are very few existing results even for capacity of point to point channels \cite{journals/tit/MedardG02,Sethuraman2009}. Our results show a behavior similar to those in \cite{journals/twc/LozanoP12,fgomezUnified}, which show that in a point-to-point non-coherent wideband channel there is a critical bandwidth occupancy, so that capacity is power limited when the bandwidth exceeds this critical value, and the critical bandwidth threshold grows with the receiver power. However, these results are only for point-to-point channels, and only qualitatively similar to the operating regimes of our network capacity scaling laws. Capacity results for multi-user non-coherent channels are limited, and the scaling exponent for the regime transitions may be different in non-coherent channels than in AWGN and coherent channels even for multiple-antenna point-to-point channels \cite{mainak2015wideband}.}
\end{itemize}

\subsection{Relation to Prior Work}

The seminal work by Gupta and Kumar \cite{kumar2000capacity} showed that the feasible rate in a dense \textit{ad-hoc} network scales as $R(n)\propto\Theta(\frac{1}{\sqrt{n}})$, where $n$ is the number of nodes\footnote{
  We use the standard $f(n)=O(g(n))$, $f(n)=\Omega(g(n))$ and $f(n)=\Theta(g(n))$ notations \cite{Knuth1976} to respectively represent that at sufficiently high $n$ function $f(n)$ becomes less than or equal than $g(n)$, greater than or equal to $g(n)$, and identical to $g(n)$ up to a constant factor.
}
. Ozgur, L\'ev\^eque and Tse introduced HC and showed that it achieves linear scaling (i.e. $R(n) = \Theta(1)$) for dense \textit{ad-hoc} networks \cite{Ozgur2007}. Franceschetti, Migliore and Minero described physical constraints which pose an ultimate limitation leading to $R(n)\leq\Theta(\frac{\log(n)^2}{\sqrt{n}})$ \cite{Franceschetti2009}. \remembertext{OzgurVSFranceschetti}{The results of \cite{Ozgur2007} and \cite{Franceschetti2009} differ in the channel model, where \cite{Ozgur2007} considers random i.i.d. phases between any pair of nodes, and \cite{Franceschetti2009} considers that as $n$ grows, the inter-node distances become smaller than the wavelength and channel phases are determined by spatial characteristics. In \cite{ozgur2009information}, Ozgur, Johary, Tse and L\'ev\^eque argue that linear scaling may still be achievable in a transitory regime where $n$ is very high but finite, such that nodes separations are larger than the carrier wavelength and channels can still be modeled by i.i.d. random phases. Otherwise, if $n$ is so high that inter-node distances are shorter than the wavelength, channel degrees of freedom scaling is spatially limited as in \cite{Franceschetti2009}. The connection between capacity scaling results with an i.i.d. random phase model and with a physical spatially-limited channel phase model is further analyzed in \cite{Ozgur2010,Lee2010,Lee2012}, which formalize the unification of \cite{Ozgur2007} and \cite{Franceschetti2009}.} In \cite{ozgur2009information} the authors also replaced the traditional separate analysis of dense and extended networks with a generalized analysis of \textit{operating regimes}, defining the user density scaling and determining its threshold for which the operating regime changes from dense-like to extended-like networks. More recently, the practicality of hierarchical cooperation to achieve linear scaling was put into question in \cite{Hong2014}. There have been extensions of scaling laws of ad-hoc networks introducing cooperation, mobility, broadcast, infrastructure or large bandwidth. See \cite{Lu2013} for a comprehensive review.

Most literature on scaling laws follows ad-hoc network models, which are not adequate representations of a cellular network, even in the case of results like \cite{Kozat2003,journals/tit/ShinJDVCLT11,wonyong2014infrastructure} that have modeled ad-hoc networks with infrastructure support. Our analysis still uses the spatial density model for infrastructure proposed in \cite{Kozat2003,journals/tit/ShinJDVCLT11,Li2011a,wonyong2014infrastructure}, but we have taken into account that data in a cellular network is \textit{required} to reach the BS and this may create bottlenecks that limit scaling \cite{Zeger2014}. \remembertext{RCShinTraffic}{In \cite{journals/tit/ShinJDVCLT11,wonyong2014infrastructure} the analysis characterizes an ``ad-hoc network with infrastructure support'', where source-destination pairs of user nodes of the same type are formed across the network, and BS infrastructure only assists these user nodes. We consider instead a conventional cellular network traffic model, where user nodes are paired with the \textit{nearest} BS, and there are typically asymmetric downlink and uplink rates with the BS as the ultimate source or destination, respectively. Note that in our multi-hop schemes nodes assist each other by forwarding information corresponding to their primary downlink/uplink exchanges with the nearest BS, however user nodes do not maintain direct traffic flows with each other using device-to-device communications underlaying the primary cellular communications, as recently proposed  \cite{BocHLMP:14,Lin2014}. Due to this cellular traffic model, our analysis requires novel cut-set bounds and achievable schemes, different than those in \cite{journals/tit/ShinJDVCLT11,wonyong2014infrastructure}.} \remembertext{RCShinAnt1}{ In addition, \cite{journals/tit/ShinJDVCLT11,wonyong2014infrastructure} considers a specific physical model for the BS antenna arrays, whereas our analysis is agnostic to the antenna model and the results are expressed instead only as a function of the \textit{effective array dimension}.}

The main innovation of our analysis method is evaluating the impact of very large bandwidths in capacity scaling. Most scaling analyses consider a constant finite bandwidth; however in such setups links only become power-limited with distance, not with bandwidth. Another approach consists on a priori letting $W\rightarrow\infty$ for each finite value of $n$, and \textit{then} let $n$ grow, as in \cite{Negi2004,Tang2008}; but this does not provide insights on the interaction between bandwidth and power-limited scaling operating regimes. In our model the goal is to find out what happens between these two extremes by letting $W$ and $n$ increase to infinity at the same time with an arbitrary relative exponent
\begin{equation}
\label{eq:defpsi}
    \psi:=\lim_{n,W\rightarrow \infty}\frac{\log W}{\log n},\; \Leftrightarrow W=\Theta(n^{\psi})
\end{equation}
where the two extremes correspond to $\psi=0$ and $\psi=\infty$. The results in \cite{ozgur2009information} also identify operating regimes depending on power scaling for the ad-hoc case, which may be interpreted implicitly as a scaling bandwidth. Introducing the bandwidth exponent explicitly allows to analyze the relative value of bandwidth scaling in relation to node and infrastructure density.

More recently, several works have studied the impact of density in cellular wireless systems with models based on stochastic geometry \cite{andrews2014selfbackhaul}. Although this permits a fine characterization of rate beyond scaling in large networks, the ability to model multi-hop protocols through stochastic geometry is more limited, and both analysis techniques are complementary. For example, in \cite{andrews2014selfbackhaul} only two hops are possible, whereas in this paper we adopt the generalized multi-hop model with arbitrary number of hops developed in the classic Gupta-Kumar model \cite{kumar2000capacity,ozgur2009information}.

\remembertext{NOlognotation}{In this paper we present scaling results characterized up to the exponents only, ignoring logarithmic variations in scaling. That is, we will not distinguish between scaling functions of the form $\Theta(n^x)$ and $\Theta(n^{x}\log(n))$. Since $n^\epsilon\geq\Theta(\log(n))$ holds for any $\epsilon\geq0$ and for sufficiently high $n$, our simplification of scaling notation does not affect the main conclusions in this paper, which mainly consider the values of rate exponents and their comparison at regime transition points. While our simplification is sufficient for the particular conclusions in this paper, we do not claim that logarithmic scaling differences are irrelevant in other applications. Researchers have committed considerable effort to study logarithmic gaps in other scenarios \cite{Franceschetti2007,Li2011a}.}

\subsection{Paper Organization}

This paper is structured as follows: Section \ref{sec:model} describes the cellular network scaling and channel model. Section \ref{sec:upperbound} obtains an upper bound to capacity. Section \ref{sec:protocols} describes the different achievable protocols. Section \ref{sec:achievabilityresults} describes capacity scaling and its relation to the throughput of each protocol. Section \ref{sec:observations} contains observations and interpretations of the results. Finally, Section \ref{sec:conclusion} concludes the paper.

\section{Network and Channel Models}
\label{sec:model}

\subsection{Network  Model}
\label{sec:network}
We consider a sequence of cellular wireless networks indexed by $n$, where $n$ is the number of single-antenna user nodes randomly and uniformly distributed in an area $A$. The network is supported by $m$ BSs, each with $\ell$ effective antennas (see below), and communication takes place over bandwidth $W$. The BSs do not have the ability to perform cooperative transmission/reception through backhaul. In the IRH protocol defined below, we also have $k>m$ single-antenna fixed RNs that communicate with the BSs through a wireless backhaul. 

Table \ref{tab:exponents} defines the scaling relation between $n$ and the different network parameters. Here $W_0$, $A_0$, $m_0$, $l_0$, $k_0$ are fixed constants. The exponents of the number of BSs and BS antennas are taken from \cite{journals/tit/ShinJDVCLT11}. The constraint $\beta+\gamma\leq1$ ensures that the number of infrastructure antennas per node does not grow without bounds. The scaling of the network area is as proposed by \cite{ozgur2009information} to model a continuum of operating regimes between \textit{dense} ($\nu=0$) and \textit{extended} ($\nu=1$) networks. We introduce the bandwidth scaling exponent $\psi$ as shown in \eqref{eq:defpsi}.
We also introduce the scaling exponent $\rho\geq\beta$ of the number of RNs for the IRH protocol defined below.

\begin{table}[!b]
 \centering
 \captionof{table}{Scaling Exponents of Network Parameters}
 \label{tab:exponents}
 \begin{tabular}{cc|p{4cm}}
  Exponent & Range & Parameter (vs. number of nodes $n$)\\ \hline
  $\psi$ & $[0,\infty)$ & Bandwidth $W=W_0n^\psi$\\
  $\nu$ & $[0,1]$ & Area $A=A_0n^\nu$\\
  $\beta$ & $[0,1]$ & No. of BSs $m=m_0n^\beta$\\
  $\gamma$ & $[0,1-\beta]$ & No. of BS antenna array effective dimensions $\ell=\ell_0 n^\gamma$\\
  $\rho$ & $[\beta,1]$ & No. of RNs $k=k_0 n^\rho$\\
 \end{tabular} 
 \end{table}
 
Note that $\ell$ is the number of \textit{effective} antenna dimensions, that is the maximum number of independent spatial dimensions over which a BS can communicate. In a rich scattering environments $\ell$ is equal to the number of physical antennas, whereas if scattering is sparse such that some physical antennas are correlated, $\ell$ represents the number of independent propagation paths that the array can exploit. By focusing on a given number of effective dimensions, our analysis can be applied, with appropriate values of $\ell$, to many antenna array architectures and even sparse propagation models in the literature, such as \cite{journals/tit/ShinJDVCLT11,Desgroseilliers2013,5673745,Samimi2014}. Hereafter, we use the term ``number of antennas'' to simply refer to the \textit{effective array dimensions} $\ell$ and represent by $\ell_\mathrm{t}$ and $\ell_\mathrm{r}$ the effective number of transmit and receive array dimensions.

\remembertext{RCShinAnt2}{Note that in \cite{journals/tit/ShinJDVCLT11,wonyong2014infrastructure} it is assumed that $\tilde{\ell}=n^{\tilde{\gamma}}$ physical antennas are uniformly and randomly located in an area $\Theta(n^{\nu-\beta})$. Certain physical characteristics of the model in \cite{journals/tit/ShinJDVCLT11,wonyong2014infrastructure} lead to a constraint in the number of exploited transmit dimensions that scales with the \textit{perimeter} of the array, $\Theta(n^{\frac{\nu-\beta}{2}})$. Therefore the effective number of transmit dimensions in \cite{journals/tit/ShinJDVCLT11,wonyong2014infrastructure} would be  $\gamma=\min(\tilde{\gamma},\frac{\nu-\beta}{2})$. In addition,  \cite{journals/tit/ShinJDVCLT11,wonyong2014infrastructure} imposes the equality requirement $\beta+\tilde{\gamma}=1$, which we have relaxed as $\beta+\gamma\leq1$ to account for any physical array mode with $\gamma\leq\tilde{\gamma}$, including but not limited to the model in \cite{journals/tit/ShinJDVCLT11,wonyong2014infrastructure}.}

We consider BSs that are placed at fixed distances of each other, dividing the network into regular hexagonal cells around each BS with radius $r_{\mathrm{cell}}$ and with asymptotically (as $n\to\infty$) $\frac{n}{m}$ nodes each, as in Fig. \ref{fig:model}. The RNs, when present, are uniformly distributed over cells and placed in a hexagonal layout within each cell. The downlink from the BS to the nodes and the uplink from the nodes to the BS operate independently in alternate time division duplex (TDD) frames. This imposes a $\frac{1}{2}$ penalty in rate but otherwise does not alter the scaling of capacity with $n$. BSs cannot receive in the downlink phase or transmit in uplink, while nodes can do both. 

\begin{figure}[!t]
 \centering
 \includegraphics[width=.9\columnwidth]{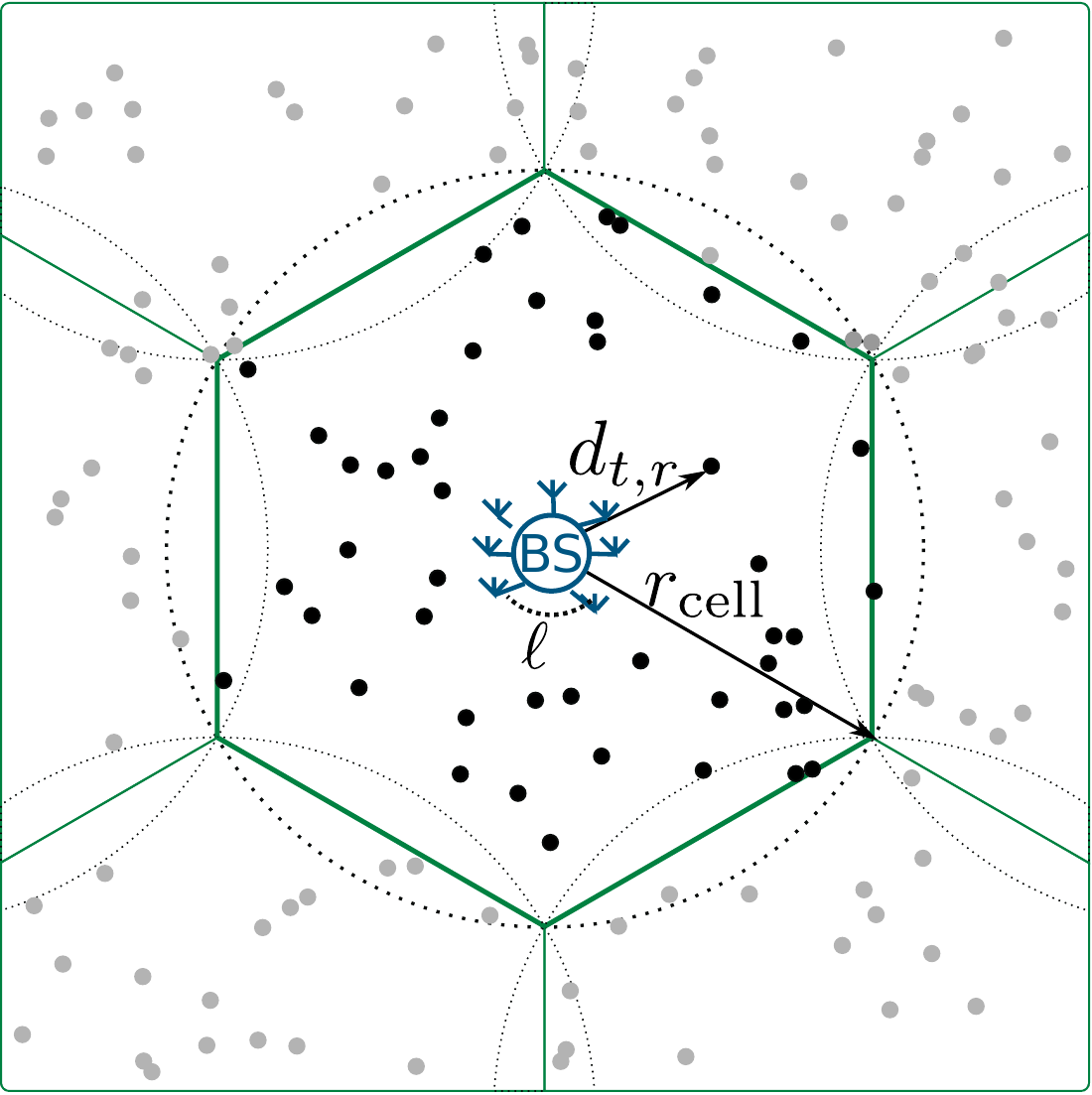}
 \captionof{figure}{One cell of the cellular network.}
 \label{fig:model}
\end{figure}

Due to random node placement, the rate achievable by any individual user is a random variable depending on its location and the protocol used. The following definitions are adapted from \cite{kumar2000capacity}.
\begin{definition}
 A downlink (uplink) rate of $R_{\mathrm{DL}}^{x}(n)$ ($R_{\mathrm{UL}}^{x}(n)$) bits per second per node is \textit{achieved using protocol $x$} in a realization of the cellular network if the protocol can guarantee that all nodes can receive from (transmit to) its assigned BS at least $R_{\mathrm{DL}}^{x}(n)$ ($R_{\mathrm{UL}}^{x}(n)$) bits per second.
\end{definition}

\remembertext{TDL}{Note that if we denote by $T_{\mathrm{DL}}^{x}(n)$ ($T_{\mathrm{UL}}^{x}(n)$) the sum DL (UL) throughput per BS with protocol $x$, our definition of achievable rate requires that $R_{\mathrm{DL}}^{x}(n)\leq \frac{m}{n}T_{\mathrm{DL}}^{x}(n)$.}

\begin{definition}
 A downlink (uplink) rate of $R_{\mathrm{DL}}(n)$ ($R_{\mathrm{UL}}(n)$) bits per second per node is \textit{feasible} in a realization of the cellular network if there exists a protocol that achieves it. In other words $R_{\mathrm{DL}}(n)={\displaystyle \sup_x} R_{\mathrm{DL}}^{x}(n)$ and $R_{\mathrm{UL}}(n)={\displaystyle \sup_x} R_{\mathrm{UL}}^{x}(n)$.
\end{definition}

The definitions above result in random rates depending on the realization of node locations. The definition below is for the largest rate scaling that holds asymptotically with probability $1$.

\begin{definition}
The downlink (uplink) per node \textit{throughput capacity scaling} $C_{\mathrm{DL}}(n)$ ($C_{\mathrm{UL}}(n)$) of random cellular network is of the order $\Theta(f(n))$ if there are constants $c_1<c_2$ such that
 \begin{equation}
 \label{eq:FeasibleP1}
 \lim_{n\rightarrow\infty}P\left(R_{\mathrm{DL}}(n)=c_1f(n)\right)=1
 \end{equation}
 \begin{equation}
 \label{eq:Unfeasible}
 \lim_{n\rightarrow\infty}P\left(R_{\mathrm{DL}}(n)=c_2f(n)\right)<1
 \end{equation}
\end{definition}

We can also define the \textit{achievable rate scaling of protocol $x$} if in the above definition we replace the feasible rates with achievable rates using protocol $x$.

In this paper we find an upper bound and lower bounds to throughput capacity scaling by studying achievable rate scaling of different protocols. When the two have the same exponent they give the capacity scaling.

%

\subsection{Channel Model}
\label{sec:channelmodel}

The discrete time received signal observed at a receiver $r$ which can be either a node, BS or RN, is given by 
 \begin{equation}
 \label{eq:tchannel}
  \y_r=d_{t,r}^{-\frac{\alpha}{2}}\Hb_{t,r}\x_t+\sum_{i\in\mathcal{I}}d_{i,r}^{-\frac{\alpha}{2}}\Hb_{i,r}\x_i+\z_r
 \end{equation}
where $\x_t$ is the signal of  the intended transmitter and the set $I$ refers to interfering transmitters active at the same time and over the same frequency band. Furthermore, $d_{j,r}$, $j=t$ or $j \in I$, is the distance between transmitter $j$ and receiver $r$, $\alpha$ is the path loss exponent and $\z_r\sim \mathcal{CN}(0,N_0\I_{\ell_\mathrm{r}})$ is the additive white Gaussian noise. 
 
 The channel gain matrix $\Hb\in\mathbb{C}^{\ell_\mathrm{t}\times \ell_\mathrm{r}}$ is assumed to be full rank. The full rank assumption is justified by our interpretation of $\ell_t$ and $\ell_r$ as effective antenna dimensions. Each coefficient of the channel matrix has unit gain and an arbitrary phase, $h_{t,r}^{(i,j)}=e^{j \theta_{i,j}}$, and that the channel squared norm satisfies $|\Hb|^2=\ell_{\mathrm{r}}\ell_{\mathrm{t}}$. The channel model in \eqref{eq:tchannel} is applicable to one symbol transmission with period $T_\mathrm{s}=1/W$ over a frequency-flat channel with power constraint $\Ex{}{|\x_t|^2}\leq \frac{P_t}{W}$ where $P_t$ depends on the type of the transmitter and the fraction of power it dedicates towards $r$. Average transmission power constraints of nodes, BSs and RNs are  $P$, $P_{\mathrm{BS}}$ and $P_{\mathrm{RN}}$, respectively.

\section{Upper Bound to Capacity Scaling}
\label{sec:upperbound}

In order to obtain the upper bound in Theorem 1 below, we develop a series of $m$ cut-set bounds to the downlink sum-rate of the users in each cell, by using the cut that separates the BS -as transmitter- from the rest of the network. Similarly, we develop $n$ cut-set bounds to uplink rate of each user by separating that particular user -as transmitter- from the rest of the network.

\begin{theorem}
\label{the:UB}
The downlink throughput capacity scaling of a cellular network is upper bounded 
\begin{equation}
 C_\mathrm{DL}(n)\leq\Theta\left(n^{\beta+\gamma-1+\min\left(\psi,(1-\nu)\frac{\alpha}{2}\right)}\right)
\end{equation}
and the uplink throughput capacity scaling is upper bounded by
\begin{equation}
 C_\mathrm{UL}(n\leq\Theta\left(n^{\min\left(\psi,(1-\nu)\frac{\alpha}{2}\right)}\right)
\end{equation}
\end{theorem}

\begin{proof}

We introduce the detailed analysis for downlink. Uplink follows similarly. We first consider the case of no RNs ($k=0$)  and then argue the same bound holds for any $k \leq n$.

We upper bound the sum-rate of the users served by each BS by considering a cut separating that BS from the rest of the network. Each of these $m$ cuts upper bounds the sum downlink rate received by the approximately $n/m$ destination users in one cell. At the receiving side of the cut there is perfect cooperation among $n$ receiver nodes and the remaining $m-1$ BS transmitters whose transmissions are known to the receivers and can be perfectly canceled. Hence, each cut behaves as a single MIMO channel with array dimensions $\ell_\mathrm{t}=\ell$ and $\ell_\mathrm{r}=n$.

We represent the distance from each node $r$ to BS $b$ in a diagonal matrix 
$$
\D_b\triangleq\left(\begin{array}{ccc}
           d_{b,1}^{-\frac{\alpha}{2}}&\dots&0\\
           \vdots&\ddots&\vdots\\           
           0&\dots& d_{b,n}^{-\frac{\alpha}{2}}\\
          \end{array}\right), 
          $$
          and modify the channel expression \eqref{eq:tchannel} to write the signals from all BSs to all nodes in the form
 \begin{equation}
 \label{eq:DLBoundchannel}
  \y=\D_t\Hb_t\x_t+\underset{\textnormal{known to all receivers}}{\underbrace{\sum_{t'\neq t}\D_{t'}\Hb_{t'}\x_{t'}}}+\z,
 \end{equation}
 and where $\Hb_t$ represents the channel matrix between BS $t$ and all receivers.


Using the assumption that channel matrices are full rank and $\ell\leq n/m\leq n$, following standard arguments in \cite[(7.10)]{tse2005book}, it can be shown that $T_{\mathrm{DL}}^{t}(n)$, the DL sum rate on the cell of BS $t$, can be upper bounded as
\begin{equation}
 \label{eq:rate3t}
T_{\mathrm{DL}}^{t}(n)\leq \max_{\sum_{i=1}^{\ell} P_i\leq P_{BS}} W\sum_{i=1}^{\ell}\log\left(1+\frac{P_i\lambda_i^2}{WN_0}\right)\\
\end{equation}
where $\lambda_i$ for $i\in[1,\ell]$ are the nonzero, nonnegative singular values of the matrix $\D_t\Hb_t$.

We know that 
$$\lambda_i^2\leq \sum_{i=1}^{\ell}\lambda_i^2= \tr\{\D_t\Hb_t\Hb_t^H\D_t^H\}\leq\ell\sum_{r=1}^{n}d_{t,r}^{-\alpha}.$$ 
Concavity of the logarithm suggests that $P_i^*=\frac{P_{\mathrm{BS}}}{\ell}$ maximizes this upper bound. Hence

\begin{equation}
 \label{eq:rate4t}
  T_{\mathrm{DL}}^{t}(n)\leq W\ell\log\left(1+\frac{P_{\mathrm{BS}}}{WN_0}\sum_{r=1}^{n}d_{t,r}^{-\alpha}\right)\\
\end{equation}

Notice that if $\lim_{n\to\infty}\frac{P_{\mathrm{BS}}\sum_{r=1}^{n}d_{t,r}^{-\alpha}}{WN_0}=\infty$, then the upper bound in \eqref{eq:rate4t} becomes degrees-of-freedom-limited and scales as $\Theta(W\ell)$. Conversely, if $\lim_{n\to\infty}\frac{P_{\mathrm{BS}}\sum_{r=1}^{n}d_{t,r}^{-\alpha}}{WN_0}=0$, the upper bound is power-limited and scales as $\Theta(\ell P_{\mathrm{BS}}\sum_{r=1}^{n}d_{t,r}^{-\alpha})$.

The sum $\sum_{r=1}^{n}d_{t,r}^{-\alpha}$ can be calculated using the exponential stripping method described in \cite{4294156}. Consider a series of concentric rings centered at the BS $t$ with inner radius $r_i=n^{\frac{\nu}{2}}e^{\frac{-i}{2}}$ and outer radius $r_{i-1}$. Recall that the user density scales as $n^{1-\nu}$ and network area as $n^\nu$, thus the number of nodes contained in each disc is $S_i\leq ne^{1-i}$ with high probability. Using this, we can upper bound the sum over $n$ by summing over the ring. Moreover, the smallest radius that contains one node w.h.p. is $r_{s}=\Theta(n^{1-\nu})$ so the sum ends at $i\leq \lfloor \log n\rfloor+1$. For all the outer rings $i\in[1,\lfloor \log n\rfloor]$ we can lower bound distance to the BS by the inner radius $d_{t,i}^{-\alpha}\leq r_{i-1}$. In addition the innermost disk indexed by $i= \lfloor \log n\rfloor+1$ contains one uniformly-distributed  node location, and its distance from the BS scales with $d_{t,i}^{-\alpha}=\theta(r_{s})=\Theta(n^{1-\nu})$ with high probability.
\begin{equation}\label{eq:sumexpstript}
\begin{split}
\sum_{r=1}^{n}d_{t,r}^{-\alpha}
&\leq \sum_{i=1}^{\lfloor \log n\rfloor+1}S_i r_i^{-\alpha}\\
&\leq \left[\sum_{i=1}^{\lfloor \log n\rfloor} ne^{1-i}n^{-\nu\frac{\alpha}{2}}e^{+i\frac{\alpha}{2}}\right]+en^{(1-\nu)\frac{\alpha}{2}}\\
&\leq n^{-\nu\frac{\alpha}{2}}\left[\log n e^{1+\frac{\alpha}{2}\log (n)}\right]+n^{(1-\nu)\frac{\alpha}{2}}e\\
&\leq (\log n+1)n^{(1-\nu)\frac{\alpha}{2}}e\\
\end{split}
\end{equation}
where the third inequality is due to $e^{+i\frac{\alpha}{2}}\leq e^{\max(i)\frac{\alpha}{2}}$.

Examining Table \ref{tab:exponents}, this leads to $\Theta\left(n^{\gamma+\min(\psi,(1-\nu)\frac{\alpha}{2})}\right)$. Now, by symmetry of the upper bound over all BSs, and by the definition of feasible rate as guaranteed to all users, the throughput capacity of the network is upper bounded by $C_{\mathrm{DL}}(n)\leq \frac{m}{n}\min_{t}T_{\mathrm{DL}}^{t}(n)=\Theta(n^{\beta+\gamma-1+\min(\psi,(1-\nu)\frac{\alpha}{2})})$
completing the proof of Theorem \ref{the:UB} for DL.

Note that this scaling upper bound makes intuitive sense because, with probability $1$ as $n\to\infty$, a disc with radius $\Theta(n^{\frac{\nu-1}{2}})$ around a BS contains one receiver, which combined with array gain $n^\gamma$ gives the best-case transfer of power between a single BS and the rest of the network. Also, the degrees of freedom of the cellular network cannot exceed $\Theta\left(Wm\ell\right)$.

A similar set of arguments lead to the bound for the uplink. In this case we consider $n$ cuts, each separating one user node from the rest of the network. In this cut, all the BSs and the remaining $n-1$ nodes are on the receiving side of the cut, and their mutual interference is canceled. Due to the fact that the transmitting node has a single antenna (eigenvalue), the degrees of freedom are $\Theta(W)$. The exponential stripping sum (equivalent of \eqref{eq:sumexpstript}) in this case needs to be evaluated over $\sum_{r=1}^{n-1}d_{t,r}^{-\alpha}+\ell\sum_{r=1}^{m}d_{t,r}^{-\alpha}=\Theta(n^{(1-\nu)\frac{\alpha}{2}})$ leading to an upper bound on uplink feasible rate as $\min_t T_{\mathrm{UL}}^{t}(n)=\Theta\left(n^{\min(\psi,(1-\nu)\frac{\alpha}{2})}\right)$

Note that the above scaling laws also apply to the downlink/uplink throughput capacity scaling of a network with $k<n$ RNs. This can be shown evaluating the cut-set bound on an equivalent network with $2n\geq n+k$ user nodes and multiplying the resulting rate per node by $2$ which is always greater than $\frac{n+k}{n}$. The capacity scaling exponent does not change when the number of nodes is multiplied by a constant.
\end{proof}

\section{Protocol Models}
\label{sec:protocols}

\subsection{Infrastructure Single-Hop (ISH)}
\label{sec:protocolsISH}
In the ISH protocol, BSs transmit directly to all nodes in downlink and all nodes transmit directly to the BSs in uplink. The BSs do not cooperate in transmission or reception and the interference signals between different cells are treated as noise. There are $\frac{n}{m}$ nodes uniformly distributed within each cell. The $\ell$ BS antennas are used for Multi-User MIMO (MU-MIMO), implementing a spatial multiplexing scheme to groups of $\ell$ users that allows each BS to transmit or receive $\ell$ signals per bandwidth resource at the same time.

Each BS divides the nodes in its cell in $\frac{n}{m\ell}$ groups of $\ell$ users, and assigns to each group a subchannel with bandwidth $W\ell\frac{m}{n}$. Subchannels are separated using Frequency Division Multiplexing (FDM) in DL and Frequency Division Multiple Access (FDMA) in UL. Within the subchannel of each group, $\ell$ simultaneously transmitted signals coexist using MU-MIMO spatial multiplexing as described in \cite[Ch. 10]{tse2005book}. The dimensions of the channel matrices always allow this because $\gamma<1-\beta$, so there are always more nodes than BS antennas if $n$ is sufficiently large. Also, the multi-user channel matrix, obtained by putting together all point-to-point channel matrices of nodes in the same subchannel and cell, is full rank if nodes are separated at least a quarter of a wavelength and far-field propagation holds. In downlink, the BS transmits independent signals to each destination with equal power allocation $P_{\mathrm{BS}}\frac{m}{n}$.

Note that while our equal division of power and bandwidth may be suboptimal we show, as a part of the analysis of ISH, that in a scaling law sense this power allocation suffices to get the best scaling possible with any single-hop protocol.
  
\subsection{Infrastructure Multi-Hop (IMH)}
In the IMH protocol, each cell is subdivided regularly into smaller regions of area $A_{\mathrm{r}}$ called \textit{routing subcells}, and information is forwarded to/from the BS via multi-hop communication using a node in each routing subcell as a relay as shown in Fig \ref{fig:modelIMH}. For multi-hopping, the routing subcells must have at least one node with high probability, which results in $A_{\mathrm{r}}>\frac{A}{m}\frac{2\log(\frac{n}{m})}{\frac{n}{m}}$ \cite{journals/tit/ShinJDVCLT11}. Routes are defined as successions of transmissions between adjacent subcells, where each hop covers a distance no longer than four subcell radii, $4r_{\mathrm{subcell}}\propto\sqrt{A_{\mathrm{r}}}$, given the largest distance between any two points in adjacent subcells. Sub-cells alternate in becoming active using a non-scaling (i.e. constant) time division scheduling by a constant factor to avoid collisions (transmissions to the same destination subcell) and satisfy the half-duplex constraint. For example in a hexagonal tessellation a $1/7$ constant can prevent collisions as illustrated in Fig. \ref{fig:routingTDMA}.

\begin{figure}[!t]
\centering
 \includegraphics[width=0.9\columnwidth]{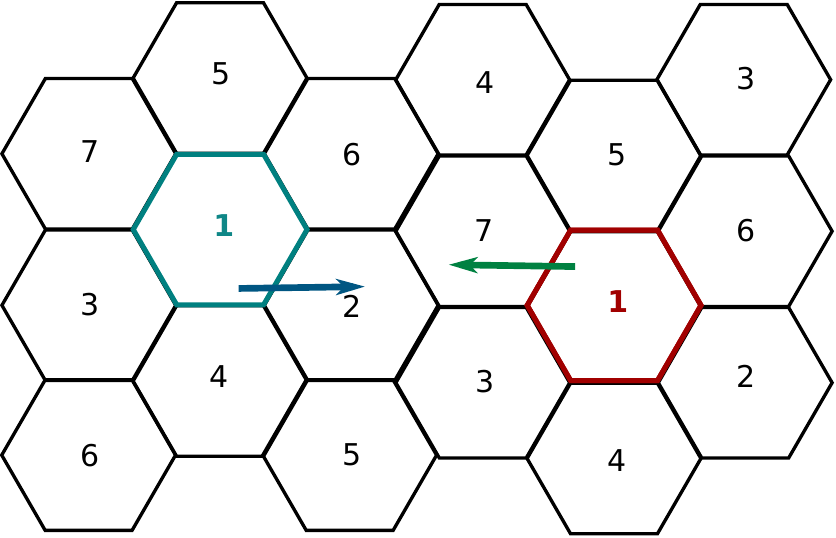}
  \caption{Time division for multi-hop routing without collisions in a hexagonal lattice. A factor of $1/7$ guarantees that two active subcells are always separated by at least two inactive subcells so that transmitters never target the same destination subcell.}
 \label{fig:routingTDMA}
\end{figure}

All downlink routes start, and all uplink routes end, at the BS in the center of the cell. We call this point the \textit{head} of all routes, where the BS communicates with its closest $\ell$ users only, using the same MU-MIMO we described for ISH with minor adaptations. Since there are $\ell$ nodes and $\ell$ BS antennas, a single channel with bandwidth $W$ without FDMA/FDM is employed, with MU-MIMO spatial multiplexing in exactly $\ell$ spatial dimensions. The channel and rate models for these links are the same as in ISH, with the new bandwidth allocation and a reduced maximum distance between the BS and the destinations scaling as $4r_{\mathrm{subcell}}=\Theta(n^{\frac{\nu-1}{2}})$.

The BS serves as head for a total of $\frac{n}{m}$ routes (one per cell user), but only $\ell$ can be spatially multiplexed at the same time, so the routes are time-multiplexed in a round robin fashion in the links between the BS and its neighbors, with each route being served a $\frac{m\ell}{n}$ portion of the time. For the remaining hops on each route, a single node in each routing subcell forwards its received data of a single path to a single node in the next routing subcell in the path. Since each node has a single antenna, there is no MIMO and all the bandwidth and node power are exploited. Again, inter-node distances scale at most as $4r_{\mathrm{subcell}}=\Theta(n^{\frac{\nu-1}{2}})$.

\begin{figure}[!t]
\centering
 \subfigure[IMH: Multi-hop is performed by one user node selected in each routing subcell.]{
 \centering
 \includegraphics[width=0.75\columnwidth]{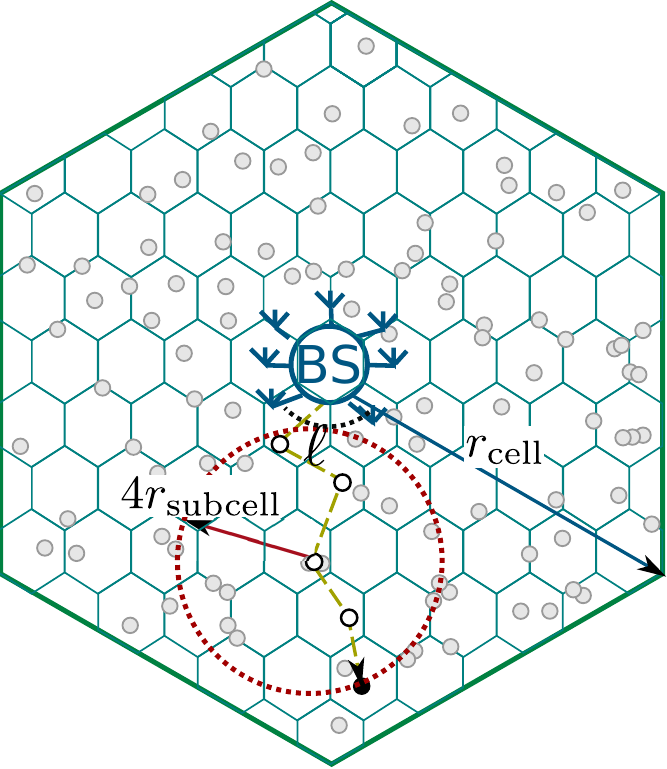}
  \label{fig:modelIMH}
 }
 \subfigure[IRH: Multi-hop is performed by RNs across microcells during the interconnection phase.]{
  \centering
  \includegraphics[width=0.75\columnwidth]{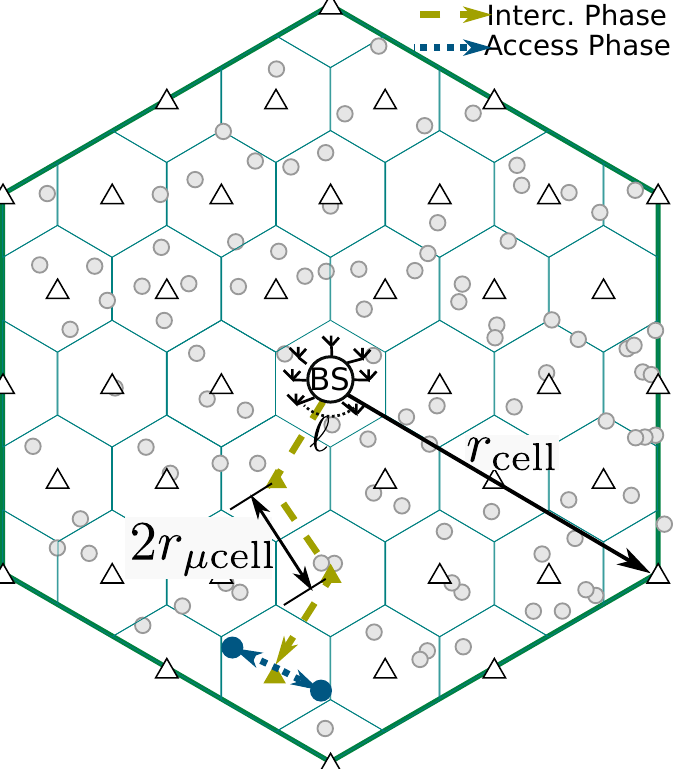}
  \label{fig:modelIRH}
  }
  \caption{The two multi-hop protocols in this paper, IMH and IRH.}

\end{figure}

\subsection{Infrastructure Relay-multi-Hop (IRH)}
In the IRH protocol, the network area is divided regularly in a nested double hexagonal grid of $m$ cells and $m+k$ microcells, where nodes in each microcell are served by an Access Point (AP) that is either a BS or a RN. We consider that a controller may decide to use the RNs or not, falling back to the behavior of ISH if the RNs are not sufficiently dense. In downlink, the RNs are exploited if $\rho\geq\beta+\gamma+(\beta-\nu)\frac{\alpha}{2}-\psi$. In uplink, the RNs are exploited if $\rho$ is high enough such that $\min\left(\psi-(\beta+\gamma-\rho)^+,(\rho-\nu)\frac{\alpha}{2}\right)\geq\min\left(\psi,(\beta-\nu)\frac{\alpha}{2}+1-\beta\right)$. These thresholds are justified by the analysis of IRH in the next section.

When the above conditions are satisfied and the RNs are utilized, a wireless backhauling connection for all $k/m$ RNs in each cell is provided by their closest BS using IMH. To implement backhauling, time is divided into an \textit{access phase} and a \textit{interconnection phase}, with relative durations $\tau_\mathrm{a}\in[0,1]$ and $1-\tau_\mathrm{a}$.

\begin{itemize}
 \item In the \textit{access phase}, for a fraction $\tau_\mathrm{a}\in[0,1]$ of the time, in each microcell, all APs exchange data with the user nodes using an ISH protocol. Signals that propagate between different microcells are treated as interference. There are $\frac{n}{m+k}$ nodes within each microcell with high probability. Unlike BSs, RNs do not have $\ell$ antennas  and therefore rates in RN microcells create a bottleneck for throughput scaling. APs use FDMA/FDM with a single antenna (no MU-MIMO), allocating transmissions to each user node on orthogonal subchannels with bandwidth $W\frac{m+k}{n}$. A BS transmits with node power allocation $P_{\mathrm{BS}}\frac{m+k}{n}$ and a RN does the same split $P_{\mathrm{RN}}\frac{m+k}{n}$. \remembertext{RC14}{Note that the capacity scaling is by definition the rate scaling of the worst user. Since microcells where the AP is a single antenna RN are more constrained, we can assume for simplicity that BSs also have a single antenna so that all microcells are represented equally in the analysis.} Note also that even if $\rho=\beta$, the numbers $m$ and $k$ may differ by a constant factor and some users in the system would be served by single-antenna APs. 

 \item In the \textit{interconnection phase}, for a fraction $1-\tau_\mathrm{a}$ of the time, BSs exchange data with RNs using an IMH protocol. Each microcell of area $A_{\mathrm{r}}\sim n^{\nu-\rho}$ becomes the \textit{routing subcell} of IMH, and information is forwarded to/from the BS via multi-hop communication using the single RN in each microcell as a relay as shown in Fig \ref{fig:modelIRH}. The BS uses MU-MIMO to transmit or receive up to $\min(\ell,\frac{k}{m})$ routing paths at the same time. Note that unlike IMH protocol, we are no longer guaranteed to have more RNs than transmit antennas. Each hop covers a distance of exactly two microcell radii, $2r_{\mu\mathrm{cell}}$, as RNs are regularly placed at the centers of their microcells. Microcells alternate in becoming active using a non-scaling (i.e. constant) time or frequency division scheduling to avoid collisions and satisfy the half-duplex constraint.
\end{itemize}

Note that the access phase and interconnection phase may have different scalings, which determines the optimal time allocation $\tau_a$ and the overall rate scaling of the IRH protocol.

We assume that RNs have one or a fixed number of antennas that do not scale with $n$. This is a realistic model since, in the near future, it is likely that BSs will still have many more antennas than nodes or RNs. However, it is not difficult to extend the results in this paper to the case where the number of RN antennas also scale with n.

\section{Capacity Scaling and Rate Scaling of Protocols}
\label{sec:achievabilityresults}

\subsection{IMH Achieves Downlink Capacity Scaling}

Our main result is the characterization of the scaling of \textit{throughput capacity} for cellular wireless networks, which is limited by the upper bound on Section \ref{sec:upperbound} and as we show below achieved by the IMH protocol. 

\begin{theorem}
 \label{the:IMH}
 For the IMH protocol, downlink rate per node scales as
 \begin{equation}
  R^{\mathrm{IMH}}_{\mathrm{DL}}(n)=
    \Theta\left(n^{\beta+\gamma-1+\min\left(\psi,(1-\nu)\frac{\alpha}{2}\right)}\right)
 \end{equation}
 and uplink rate per node scales as
 \begin{equation}
  R^{\mathrm{IMH}}_{\mathrm{UL}}(n)=\Theta(n^{\beta+\gamma-1+\min\left(\psi,(1-\nu)\frac{\alpha}{2}\right)})
 \end{equation}
\end{theorem}
\begin{proof}
 Appendices \ref{sec:ISH} and \ref{sec:IMH}.
\end{proof}

Combining upper bound in Theorem 1 and the achievable scaling in Thorem 2, we obtain the following.

\begin{theorem}
 \label{the:cap}
The downlink per node throughout capacity scales as
\begin{equation}
 C_{\mathrm{DL}}(n)= \Theta\left(n^{\beta+\gamma-1+\min\left(\psi,(1-\nu)\frac{\alpha}{2}\right)}\right)
\end{equation}
and for $\beta+\gamma=1$, the uplink per node throughut capacity scales as
\begin{equation}
 C_{\mathrm{DL}}(n)= \Theta\left(n^{\min\left(\psi,(1-\nu)\frac{\alpha}{2}\right)}\right)
\end{equation}
\end{theorem}

\begin{corollary}
IMH is optimal for downlink, that is it achieves throughput capacity scaling in downlink. For uplink, IMH achieves  rate scaling within a gap no larger than $n^{1-\beta-\gamma}$ to capacity, and is optimal for $\beta+\gamma=1$
\end{corollary}

\begin{remark}
In our model we have $\beta+\gamma\leq 1$. When scattering is rich and the scaling of the number of independent transmits dimensions $\gamma$ equals the number of physical antennas the condition $\beta+\gamma=1$ corresponds to the total number of infrastructure investment scaling as the number of users. This, for example, is within the realm of ultra-dense networks \cite{Baldemair2015}. In this case IMH is optimal (in terms of throughput capacity scaling) in the uplink as well.
\end{remark}

Next, we obtain the achievable rate scalings of the other protocols introduced in Sec. \ref{sec:protocols}.

\subsection{ISH is Suboptimal}

The ISH protocol is representative of the dominant communication mode in current cellular networks, consisting of direct transmissions between BS and nodes. Our analysis shows that single-hop protocols can not fully exploit large bandwidths, and therefore cellular architectures must adopt multi-hop in future generations if large bandwidths are to be utilized optimally.

\begin{theorem}
 \label{the:ISH}
 For the ISH protocol, downlink rate per node scales as
 \begin{equation}
 \label{eq:ISHDLR}
  R^{\mathrm{ISH}}_{DL}(n)=
    \Theta\left(n^{\beta+\gamma-1+\min\left(\psi,(\beta-\nu)\frac{\alpha}{2}\right)}\right)
 \end{equation}
 and uplink rate scales as
 \begin{equation}
 \label{eq:ISHULR}
  R^{\mathrm{ISH}}_{UL}(n)=
    \Theta\left(n^{\beta+\gamma-1+\min\left(\psi,(\beta-\nu)\frac{\alpha}{2}+(1-\beta)\right)}\right)
 \end{equation}
\end{theorem}
\begin{proof}
 Appendix \ref{sec:ISH}.
\end{proof}
%
\begin{remark}
For $\beta=1$ ISH has the same rate scaling as IMH in all regimes and is optimal. In downlink, for all $\beta<1$ we have $(\beta-\nu)\frac{\alpha}{2}<(1-\nu)\frac{\alpha}{2}$ and ISH achieves a rate scaling worse than IMH when the bandwidth scaling is $\psi\geq(\beta-\nu)\frac{\alpha}{2}$. Similarly, in uplink, for all $\beta<1$ we have $(\beta-\nu)\frac{\alpha}{2}+(1-\beta)<(1-\nu)\frac{\alpha}{2}$ and ISH performs worse than IMH for $\psi\geq(\beta-\nu)\frac{\alpha}{2}+(1-\beta)$.
\end{remark}


\subsection{IRH Performance Depends Critically on RN Density}

There may be practical issues related with multi-hop implementation through mobile users as in IMH, and the use of static dedicated RNs in this protocol provides a reasonable middle-ground. The gap between IRH and IMH gets smaller and can be closed as RN density increases.

\begin{theorem}
 \label{the:IRH}
 For the IRH protocol, if $\rho\geq\beta+\gamma+(\beta-\nu)\frac{\alpha}{2}-\psi$, downlink rate per node scales as
\begin{equation}
 R^{\mathrm{IRH}}_{\mathrm{DL}}(n)= \Theta\left(n^{\min(\beta+\gamma,\rho)-1+\min\left(\psi,(\rho-\nu)\frac{\alpha}{2}\right)}\right)
\end{equation}
 and if $\min\left(\psi-(\beta+\gamma-\rho)^+,(\rho-\nu)\frac{\alpha}{2}\right)\geq\min\left(\psi,(\beta-\nu)\frac{\alpha}{2}+1-\beta\right)$, uplink rate per node scales as
\begin{equation}
 R^{\mathrm{IRH}}_{\mathrm{UL}}(n)= \Theta\left(n^{\min(\beta+\gamma,\rho)-1+\min\left(\psi,(\rho-\nu)\frac{\alpha}{2}+(\beta+\gamma-\rho)^+\right)}\right),
\end{equation}
 otherwise rates scale as in ISH.
\end{theorem}
\begin{proof}
 Appendix \ref{sec:IRH}.
\end{proof}

\begin{remark}
The IRH controller always uses the RNs in downlink if $\rho>\beta+\gamma$, and in uplink if $\rho\geq\beta+\max(\gamma,(1-\beta)\frac{2}{\alpha})$.
\end{remark}
\begin{remark}
If $\rho=1$ the rate scaling with IRH matches the rate scaling with IMH and therefore IRH is optimal in downlink. If $\rho=\beta+\gamma=1$ IRH is optimal in uplink as well.
\end{remark}
\begin{remark}
However if $\rho=1$ and $\beta+\gamma<1$ IRH does not meet the uplink upper bound rate scaling. This means that the amount of wired-backhauled infrastructure is a limiting factor even for networks with very high density of wireless-backhauled infrastructure. Finally, if $\rho<1$, IRH rate scaling is dominated by that of IMH.
\end{remark}

\subsection{Illustration of the results}

\begin{figure}[!t]
\centering
\subfigure[Downlink]{
 \centering
 \raisebox{.5in}{\includegraphics[width=\columnwidth]{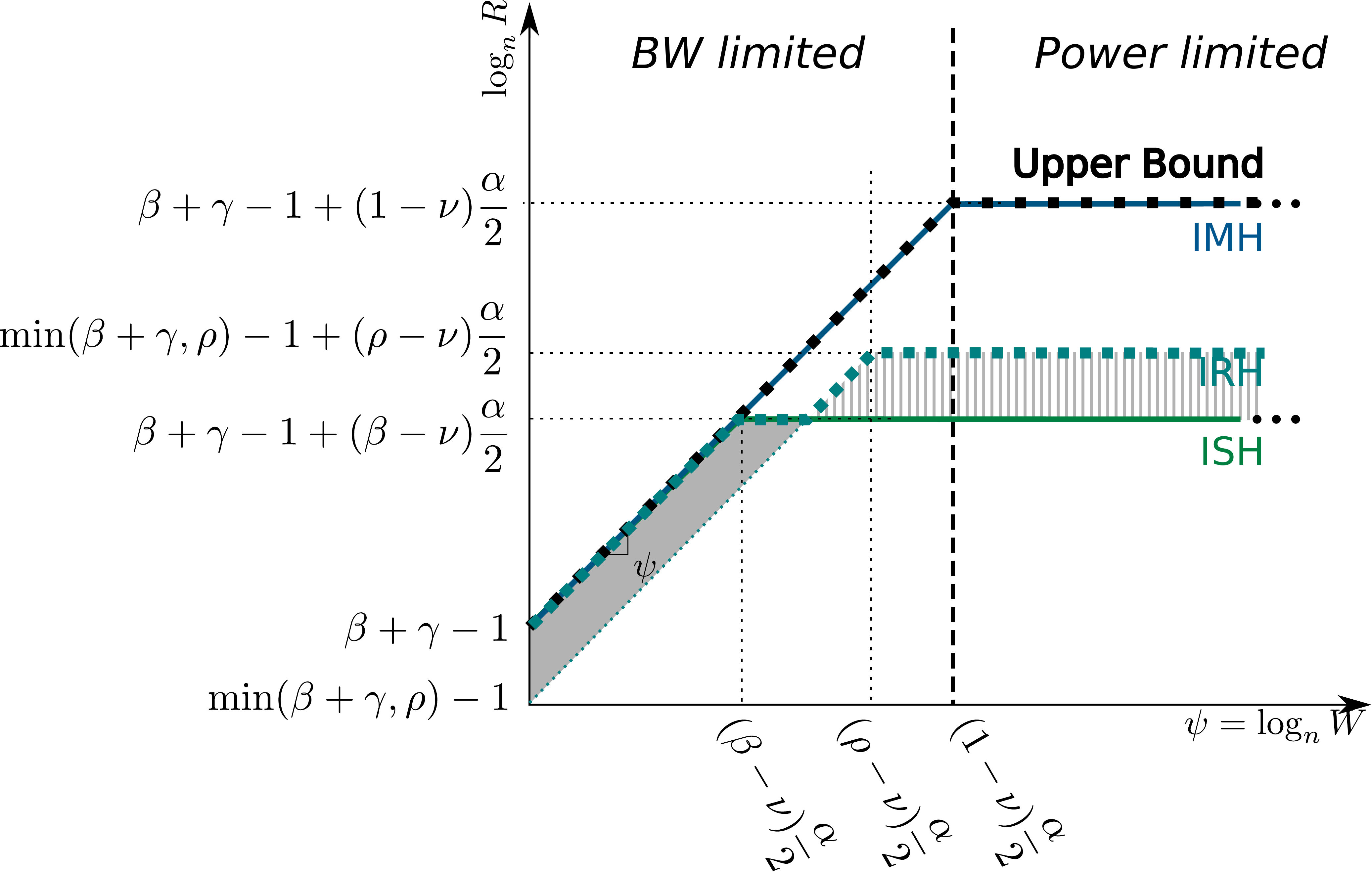}}
 \label{fig:epsilonDL}
}

\subfigure[Uplink]{
 \centering
 \includegraphics[width=\columnwidth]{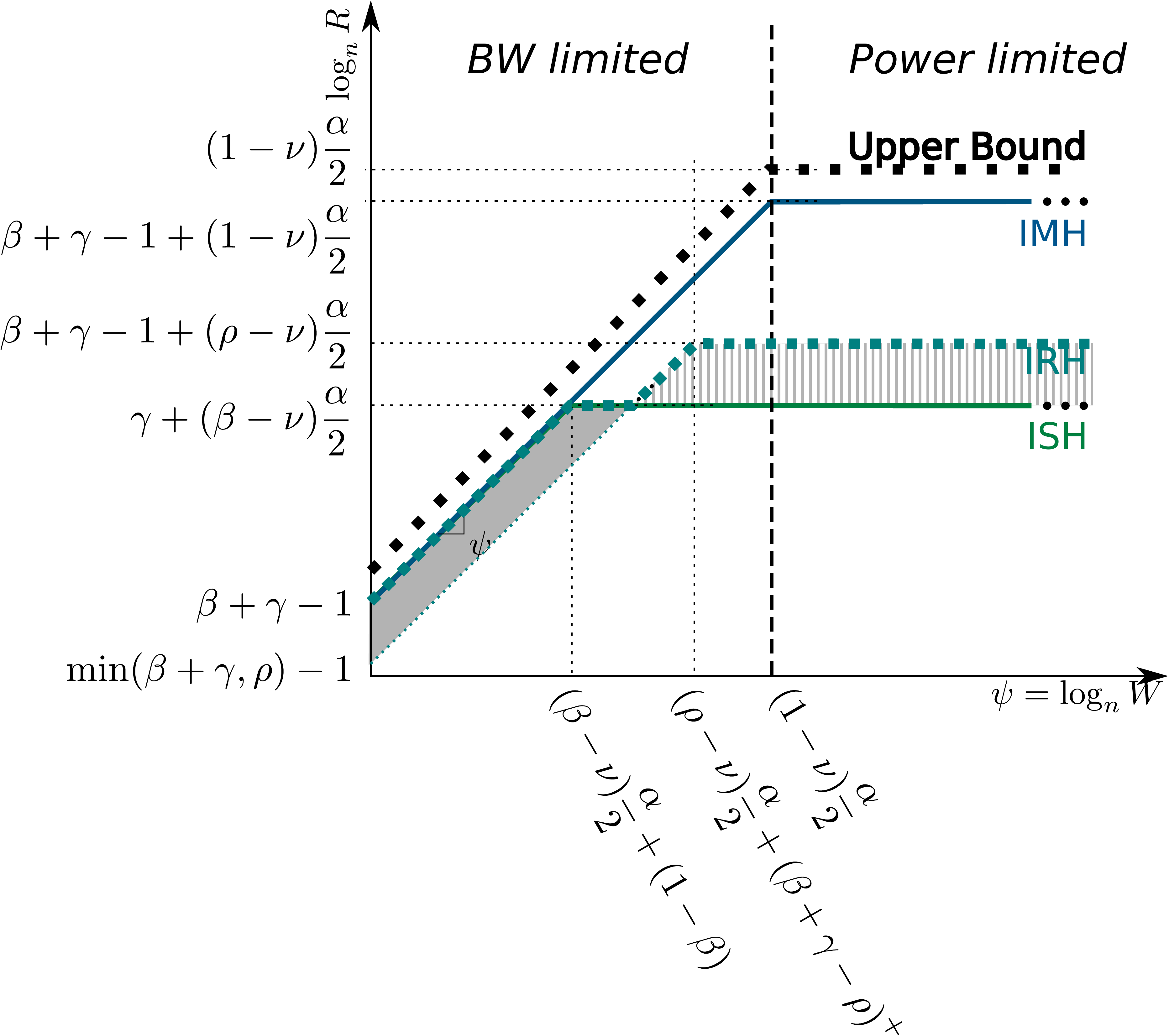} 
 \label{fig:epsilonUL}
}
 \caption{Scaling exponents for the upper bound on capacity and feasible rate with ISH, IRH and IMH. Donwlink capacity is fully characterized, whereas uplink capacity lies in the small gap between the upper bound and the IMH protocol.}
 \label{fig:epsilon}
\end{figure}
\remembertext{redundants}{
Figures \ref{fig:epsilonDL} and \ref{fig:epsilonUL} illustrate the scaling exponents of the upper bound, IMH, ISH and IRH protocols, for the downlink and uplink cases respectively. The horizontal axes represent the exponent of bandwidth, $\psi$, which together with the exponent of number of independent array dimensions, $\gamma$, represents the scaling of the degrees of freedom of BS transmission. The vertical axes show the exponent of the feasible per node rate $\log(R(n))$. 

In  downlink, the upper bound behaves exactly like IMH, and hence the capacity scaling of the cellular network is fully characterized and IMH is optimal, whereas in uplink IMH cannot achieve the upper bound if $\gamma+\beta< 1$, with the gap between IMH and the upper bound scaling as $n^{1-\gamma-\beta}$. \remembertext{RCShinAnt3}{These results generalize the particular case from  \cite{journals/tit/ShinJDVCLT11} where the bandwidh does not scale ($\psi=0$) and the antenna array at the BS follows a physical model with $\tilde{\ell}=n^{\tilde{\gamma}}$ physical antennas that experience a constraint on the number of independent transmit dimensions given by $\gamma=\min(\tilde{\gamma},\frac{1-\beta}{2})$.}

In ISH transmissions have to cover longer distances, from the BS to the cell edge, resulting in an earlier transition into a power-limited regime and a lower utility of increasing the bandwidth compared to IMH. In UL this is partially compensated by the fact that the BS receives the total power transmitted by $n/m$ nodes. However, since for a $\beta\leq1$, we have $(\beta-\nu)\frac{\alpha}{2}+(1-\beta)\leq(1-\nu)\frac{\alpha}{2}$, and ISH rate scaling is dominated by that of IMH.}

The analysis of IRH in Theorem \ref{the:IRH} shows that both in downlink and uplink we can identify a minimum density of RNs such that, beyond this density, IRH outperforms ISH. 
\begin{itemize}
\item A first scenario where the IRH controller \textbf{must not} use the RNs is identified in the switching conditions on Theorem \ref{the:IRH}. We have highlighted the region where RNs must not be used with a solid gray shadowed area in the figures. This gap region exists only if $\rho\leq\beta+\gamma$, whereas if $\rho>\beta+\gamma$ the rate scaling with bandwidth exponent $\psi\leq(\beta-\nu)\frac{\alpha}{2}$ is unchanged regardless of whether the controller use the RNs or not.
\item A second scenario where the IRH controller \textbf{must} use the RNs is identified by comparing the rates of IRH and ISH at $\psi>(\beta-\nu)\frac{\alpha}{2}+(\beta+\gamma-\rho)^+$. We have highlighted this gap with a striped gray shadowed area in the figures. The received power defines a bottleneck when $\psi>(\beta-\nu)\frac{\alpha}{2}+(\beta+\gamma-\rho)^+$ in downlink and $\psi>(\beta-\nu)\frac{\alpha}{2}+(1+\beta)+(\beta+\gamma-\rho)^+$ in uplink. RNs introduce new bandwidth-limited and power-limited areas to the rate exponents, and allow IRH to outperform ISH.
  In downlink this gap always exists for any $\rho\geq\beta$, and RNs always increase power-limited rates, because the receiver power depends only on the distance and pathloss exponents, which are improved by RNs to  $(\rho-\nu)\frac{\alpha}{2}>(\beta-\nu)\frac{\alpha}{2}$. 
  In uplink, however, if $\rho$ is too low, the striped gray region denoted in the figure collapses and RNs do not improve the power-limited rates of ISH. This occurs because the power bottleneck depends also on the number of uplink transmitters per receiver, and RNs only improve the received power if $(\rho-\nu)\frac{\alpha}{2}+(\beta+\gamma-\rho)^+>(\beta-\nu)\frac{\alpha}{2}+(1-\beta)$.
\end{itemize}

\section{Discussion and Observations}
\label{sec:observations}
\subsection{The limitations of ISH are fundamental to any single-hop protocols}

\remembertext{RC16}{The ISH protocol defined in Section \ref{sec:protocolsISH} assumes equal power allocation in downlink and suboptimal linear MU-MIMO processing. It is not difficult to show that a more general version of single hop would not improve the rate scaling beyond what was obtained in Theorem 4. Uplink and downlink rates in a general cellular network restricted to single-hop communications can be upper bounded by considering broadcast and multiple access channel results \cite{tse2005book}, respectively. For example, our downlink analysis, which is asymptotic in number of nodes, can be identified as an asymptotic high-SNR broadcast channel when $\psi\leq(\nu-\beta)\frac{\alpha}{2}$. The degrees of freedom region of the broadcast channel in the high-SNR regime is known (see for example \cite{tse2005book}) and the worst user performance does not exceed --in terms of scaling exponent-- what we achieve with ISH. Similar arguments can be made regarding the low-SNR capacity of the broadcast channel with $\psi\geq(\nu-\beta)\frac{\alpha}{2}$, and the multiple access channel. As a result, no other single-hop protocol throughput scales better than that of ISH, and the differences between ISH and IMH in our achievable schemes arise from the differences between the single-hop and multi-hop architectures, not from our simplifications in ISH. }

\subsection{Operating regimes of large cellular networks}

\begin{figure*}[!t]
 \centering
\subfigure[Network bandwidth-limited regime type I.]{
 \includegraphics[width=.315\textwidth]{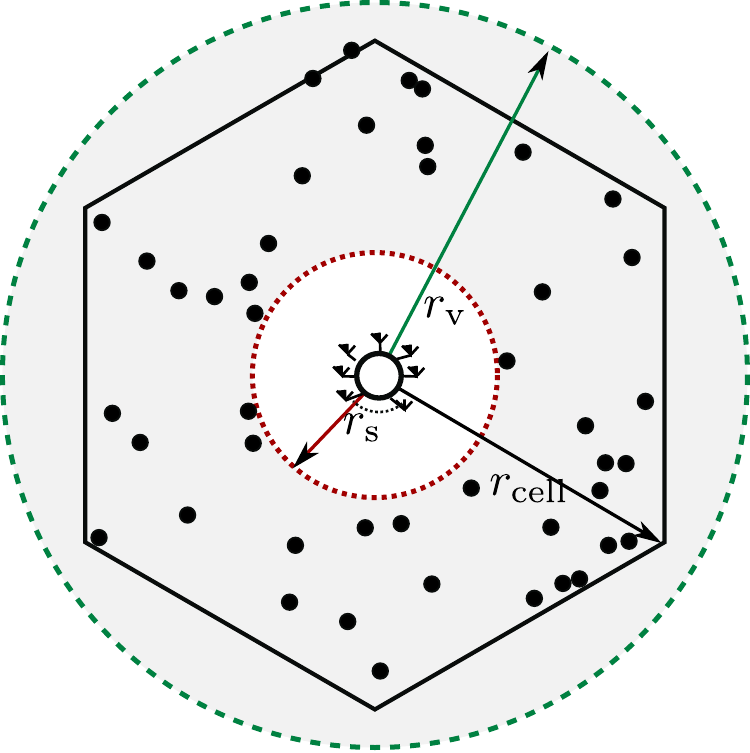}
 \label{fig:regime1}
 }
 \hspace{.031in}
\subfigure[Network bandwidth-limited regime type II.]{
 \includegraphics[width=.245\textwidth]{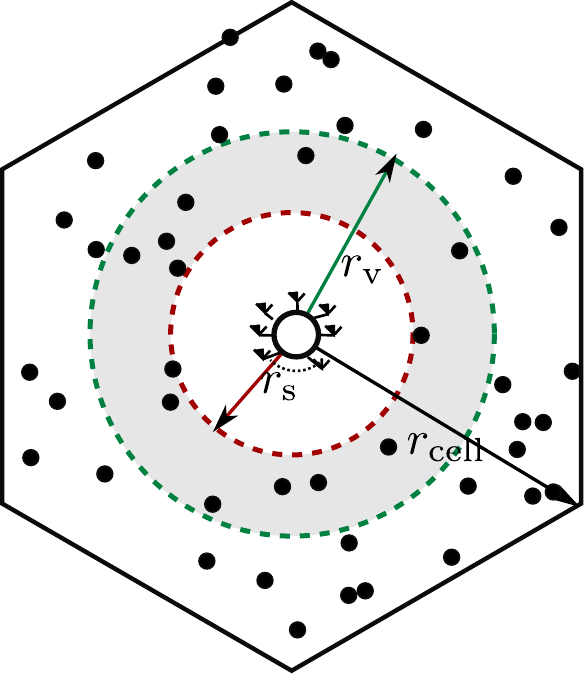}
 \label{fig:regime2}
 }
 \hspace{.3in}
\subfigure[Network power-limited regime.]{
 \includegraphics[width=.245\textwidth]{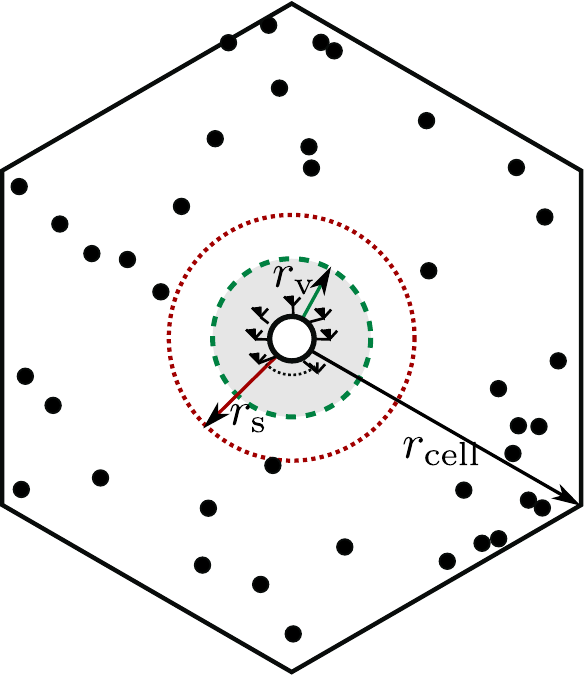}
 \label{fig:regime3}
 }
 \caption{Capacity scaling regimes. In the \textit{network bandwidth-limited regime type I} all nodes lie sufficiently close to the BS to receive direct high-SNR transmissions. In the \textit{network bandwidth-limited regime type II} capacity is bandwidth-limited, but most nodes lie far from the BS and cooperation among nodes is imperative. In the \textit{network power-limited} regime capacity is power-limited.}
 \label{fig:regimes}
 \vspace{-.2in}
\end{figure*}

For point to point channels operating with power $P$, bandwidth $W$, it is well known that there are two operating regimes: When $\frac{P}{WN_0}\ll1$, the capacity, given by $C(W)=W\log(1+\frac{P}{WN_0})$. behaves as $\Theta(\frac{P}{N_0})$, and we say that it is \textit{power limited}. Conversely, when $\frac{P}{WN_0}\gg1$, the capacity behaves as $\Theta(W\log(\frac{P}{N_0}))$, and we say it is \textit{bandwidth limited}.

Our analysis shows that large cellular networks also have two capacity scaling regimes: {\em Network bandwidth limited regime} and the {\em network power limited regime.} Furthermore, the network bandwidth limited regime can be categorized into two types depending on whether cooperation among nodes is necessary or not to ensure that network power is not a limitation. 

We illustrate these regimes in figure \ref{fig:regimes}. The regimes apply to both downlink and uplink, but we describe them only for downlink for the sake of compactness. We denote the cell radius by $r_{\mathrm{cell}}=\Theta(n^{\frac{\nu-\beta}{2}})$, the distance between two closest nodes by $r_{\mathrm{s}}=\Theta(n^{\frac{\nu-1}{2}})$, and the longest distance where the BS can transmit or receive without being power limited by $r_{\mathrm{v}}=(W)^{-\frac{1}{\alpha}}$.

\textit{Network bandwidth-limited regime type I:} If $\psi<(\beta-\nu)\frac{\alpha}{2}$, $r_{\mathrm{v}}$ scales faster than $r_\mathrm{cell}$, as $n\to\infty$, and it is possible to deliver bandwidth-limited rates $\Theta(n^{\beta+\gamma-1+\psi})$ separately to each node in the cell using single hop protocols such as ISH. In this regime there is no requirement for cooperation.

\textit{Network bandwidth-limited regime type II:} For $(\beta-\nu)\frac{\alpha}{2}<\psi\leq(1-\nu)\frac{\alpha}{2}$ network capacity is still bandwidth limited, but single-hop protocols are not. The radius $r_{\mathrm{v}}$ scales slower than $r_\mathrm{cell}$ but faster than $r_{\mathrm{s}}$. A few nodes in each cell are sufficiently close to their BS to establish high-SNR direct communications, while a majority of nodes are further away at low-SNR distances. Therefore, in this regime multi-hop is imperative to achieve the bandwidth-limited capacity scaling $\Theta(n^{\beta+\gamma-1+\psi})$.
  
\textit{Network power-limited regime:} If $\psi>(1-\nu)\frac{\alpha}{2}$,  $r_{\mathrm{cell}}$ grows faster than $r_{\mathrm{v}} $ and we cannot guarantee that there is at least one user sufficiently close to the BS with high probability. In this regime the SNR in expression \eqref{eq:rate4t} is low and the upper bound to capacity is power-limited ($\Theta(n^{\beta+\gamma-1+(1-\nu)\frac{\alpha}{2}})$).

If $\beta=1$ ISH is optimal, the type II bandwidth-limited regime collapses and the rates of ISH and IMH scale with the same exponent. However, the constraint on the total number of infrastructure units per user $\beta+\gamma\leq 1$ means that $\beta=1$ is incompatible with the exploitation of large antenna degrees of freedom in the BSs.

\subsection{Relation to the history and future of cellular technologies}

Current cellular networks are limited in degrees-of-freedom, namely bandwidth and number antennas. Hence they operate only in the network bandwidth limited regime type I, where user nodes are individually bandwidth limited and single hop protocols are optimal. Therefore, it is not surprising that the gains obtained by the early implementations of relaying in 4G systems have been modest \cite{fgomez2014improvedrelaying}.

These early multi-hop implementations correspond to our IRH protocol with low relay density (low $\rho$). In broad terms, multi-hopping uses degrees-of-freedom in exchange of power gain, and therefore is not as advantageous in traditional bandwidth limited cellular networks such as LTE, where degrees of freedom are a precious resource \cite{fgomez2014improvedrelaying}. However, future cellular systems, such as mmWave, will have an abundance of bandwidth and most likely operate in the network power-limited regime, necessitating multi-hopping, either using network nodes or infrastructure RNs,  for increased network capacity \cite{7499308}. 

Interference studies have shown that mmWave networks operate very close to the threshold between having interference- and noise-limited links \cite{7499308} (roughly equivalent to network bandwidth- and power-limited rates). The analys in \cite{7499308} also showed that a small increase in node density makes the mmWave network transition from noise- to interference-limited. Therefore, the operating regime transitions identified in our analysis have a direct practical impact in the design of future mmWave cellular systems.

\subsection{IRH is more practical than ISH or IMH for high density networks }

Both ISH and IRH improve with large invesments in infrastructure, increasing $\beta$ or $\rho$, respectively, but in practice IRH offers advantages over ISH because there are practical limitations to the deployment of wired BSs to increase $\beta$, such as curb excavation rights and cost of fiber-optic for backhaul connection.

IRH also has practical advantages over IMH as a multi-hop implementation, due to the fact that RNs are typically static, connected to the energy grid, and owned by same network operator. Conversely, user nodes are typically mobile, battery powered and owned by customers. Thus the implementation of IMH would pose greater practical challenges such as support for  mobility in the multi-hop protocol, battery efficiency optimization, and behavioral incentives to prevent customers from rejecting cooperation.

\subsection{Hierarchical cooperation is not necessary for cellular capacity scaling}

In an ad-hoc network both direct transmission and multi-hop may be suboptimal in some regimes, and a HC protocol is employed to achieve capacity scaling \cite{ozgur2009information}. This demonstrates the utility of cooperative \textit{virtual antenna arrays}, formed through coordinated joint transmissions by single antenna nodes grouped in clusters. Our analysis of cellular networks shows instead that IMH achieves downlink capacity scaling in all regimes. Therefore node clusters forming virtual antenna arrays are not necessary to achieve capacity scaling in cellular networks.

We leave for future work a scenario in which HC might regain relevance for cellular networks, where we relax the assumption that BSs cannot exchange their messages through a backhaul connection and perform joint transmissions. For cooperative BSs, it is possible that virtual antenna arrays formed by clusters of cooperative devices become necessary at the user side.

\subsection{Cellular-specific traffic bottlenecks affect capacity scaling}

Some existing ``ad-hoc networks with infrastructure support'' scaling laws analyses \cite{Kozat2003,journals/tit/ShinJDVCLT11,wonyong2014infrastructure,Dhillon2014backhaul,dhillion2014scalability,Yoon2014} model infrastructure only as an intermediary to assist the same ad-hoc type communications. These works study rates from some user nodes to others, where the infrastructure is a mere intermediary, and when BSs are so far apart that communicating with them does more harm than good, these works ignore the infrastructure and apply ad-hoc protocols such as HC. The more realistic cellular network analyzed in this paper requires traffic to always flow through BSs even when this causes bottlenecks as illustrated in our analysis.

\section{Conclusions}
\label{sec:conclusion}

In this paper we have obtained the throughput capacity scaling of cellular wireless networks in a model comprising scaling of area, BS density, number of antennas per BS, total bandwidth, and also, optionally, number of wireless backhauled RNs. We have shown that cellular network capacity scaling exhibits a transition between network bandwidth-limited and network power-limited operating regimes as the bandwidth increases, equivalent to the well known transition in point-to-point links from bandwidth-limited to power-limited capacity. Moreover, we have shown that different protocols can experience protocol-specific suboptimal transitions into power-limited behavior earlier than (i.e. for bandwidth scaling exponent lower than) the transition experienced by the capacity scaling. The transition thresholds are fundamentally related to the typical distance between a transmitter and a receiver for each protocol, suggesting that cooperative multi-hop schemes transmitting between nearest neighbor across the minimal distance have an advantage in networks with wide bandwidths. In fact, our results show multi-hop is optimal for downlink.

Single-hop protocols deal with the longest transmission distances and transition into power limitation the earliest. This means the network bandwidth-limited capacity regime is further divided into two subtypes: type I, where bandwidth is low enough that single-hop protocols are bandwidth-limited and all users can be served \textit{independently} with bandwidth-limited rates; and type II, where the network capacity is bandwidth-limited but single hop protocols are power-limited and cooperation is imperative to serve capacity-achieving rates to all users.

In cellular networks with additional wireless-backhauling RNs the capacity scaling law depends strongly on RN density, so that if the number of RNs is insufficient the network is better off disregarding the RNs altogether and using only the BSs. Conversely, if a sufficiently dense set of RNs is installed, cooperative multi-hop capacity can be achieved. However, a low density of wired-backhauling infrastructure can still be a limitation even if the RN density is high.

Our analysis provides a theoretical framework to explain historical experiences with the implementation of multi-hop and relaying in cellular networks, where the gains have been very modest. Current cellular systems operate in the network bandwidth-limited type I regime, where our analysis predicts that adding RNs brings little advantage and may even decrease rate scaling.

Our analysis is also highly relevant for the design of future cellular networks with increased bandwidth, such as mmWave or carrier aggregation systems. Preliminary studies in mmWave have shown highly parameter-sensitive transitions from bandwidth-limited to power-limited behavior, suggesting potentials for multi-hop communications.

\appendices
\section{Proof of Theorem \ref{the:ISH}}
\label{sec:ISH}
We start with the analysis of ISH, which also forms the foundation of the IMH and IRH protocols. We describe the downlink proof in detail, whereas the uplink proof follows by minor changes. In MU-MIMO downlink BS $t$ assigns to each user node $r$ in the same subchannel a signature unitary vector $\uu_{t,r}$ and transmits $\x=\sum_{r=1}^{\ell}\uu_{t,r}x_{t,r}$, satisfying the power allocation $\Ex{}{|\uu_{t,r}x_{t,r}|}=P_{\mathrm{BS}}\frac{m}{n}$ by the protocol description of ISH (Sec. \ref{sec:protocols}). Out-of-cell interfering transmissions in the same subchannel and transmissions by the BS not canceled by the linear precoder are treated as noise:
 \begin{equation}
 \label{eq:ISHchannel}
 \begin{split}
  y_{t,r}=&d_{t,r}^{-\frac{\alpha}{2}}(\h_{t,r}^H\uu_{t,r})x_{t,r}+\underset{I_1 \textnormal{ (Same BS)}}{\underbrace{\sum_{r'}d_{t,r}^{-\frac{\alpha}{2}}(\h_{t,r}^H\uu_{t,r'})x_{t,r'}}}\\
  &+\underset{I_2 \textnormal{ (Other BSs)}}{\underbrace{\sum_{(t',r')\in\mathcal{I}^{\mathrm{ISH}}_{t,r}}d_{t',r}^{-\frac{\alpha}{2}}(\h_{t',r}^H\uu_{t',r'})x_{t',r'}}}+z_r
  \end{split}
 \end{equation}
 
We represent by $I_1$ the self-interference of the signals transmitted by the same BS towards other nodes, and by $I_2$ the out-of-cell interference from the set of interferers $\mathcal{I}^{\mathrm{ISH}}_{t,r}$ consisting in all transmitter-receiver pairs $(t',r')$ allocated to the same subchannel as $(t,r)$ by some other BS $t'$. The design of the transmit vectors $\uu_{t,r}$ is studied in \cite[10.3]{tse2005book}. We apply only linear transmit precoding vectors $\uu=\h/\sqrt{\ell}$ that are optimal at low-SNR and suboptimal at high-SNR. However, since we are only interested in scaling, less power efficiency in the bandwidth-limited regime does not affect the scaling exponent. This makes the channel gain at the desired user
$$|\h_{t,r}^H\uu_{t,r}|^2=\left|\frac{\h^H\h}{\sqrt{\ell}}\right|^2=\ell$$

Due to the fact that all transmissions are uniformly allocated across the available bandwidth and spatial dimensions, and there are many interferers, the second form of interference can be approximated as white noise by the central limit theorem. Since Gaussian is the worst distribution a noise of known covariance can have \cite{cover2001noise}, we  lower bound the rate by also modeling the self-interference within the MU-MIMO scheme as Gaussian with variance $\Ex{}{|I_1|^2}$. The variance of the total MU-MIMO self-interference can be characterized as

$$\Ex{}{|I_1|^2}=\Ex{}{|\sum_{r'}d_{t,r}^{-\frac{\alpha}{2}}(\h_{t,r}^H\uu_{t,r'})x_{t,r'} |^2}=\Theta(\ell\frac{m}{n}P_{BS})d_{t,r}^{-\alpha}$$

For the out-of-cell interference, we introduce a notation that applies to all protocols. For protocol x we denote the noise plus out-of-cell interference Power Spectral Density (PSD) for transmitter $t$ and receiver $r$ as
 \begin{equation}
\label{eq:NIx}
N_\mathrm{I}^{x}\triangleq\frac{\Ex{}{|I_2+z_r|^2}}{W\frac{m}{n}\ell} =\frac{\sum_{t'\in\mathcal{I}^{x}_{t,r}}d_{t',r}^{-\alpha}P_{t'}}{W}+N_0,
\end{equation}
where $\mathcal{I}^{x}(t,r)$ denotes a set of out-of-cell interference transmitters affecting $t$. Here $P_{t'}$ denotes the total power of interferer $t'$. In our particular case of ISH downlink $\mathcal{I}(t,r)$ is the set of all BSs except $t$, and $P_{t'}=P_{BS}$.
 
Therefore the rate in the ISH downlink link $(t,r)$ can be lower bounded as
\begin{equation}
\label{eq:mutinfISH}
R_{t,r}^{\mathrm{ISH}}\geq \frac{m}{n}\ell W\log\left(1+\frac{\ell\frac{m}{n}P_{BS}d_{t,r}^{-\alpha}}{\Ex{}{|I_1|^2}+W\frac{m}{n}\ell N_{\mathrm{I}}^{\mathrm{ISH}}} \right),
\end{equation}

The achievable rate scaling must by definition be guaranteed to all the nodes in the network, so that $R^{\mathrm{ISH}}(n)=\min_{t,r}R_{t,r}^{\mathrm{ISH}}$ where each $R_{t,r}^{\mathrm{ISH}}$ is lower bounded by \eqref{eq:mutinfISH}. Also, \eqref{eq:mutinfISH} decreases monotonously with $d_{t,r}$, and the worst-case distance satisfies $\max_{t,r}d_{t,r}=\Theta(r_{\mathrm{cell}})$ with high probability. Therefore
\begin{equation}
\label{eq:mutinfISHregimes}
 R^{\mathrm{ISH}}(n)\geq\begin{cases}
                      \Theta(\frac{m}{n}\ell W)&\textbf{iff }r_{\mathrm{cell}}^{-\alpha}\gtrsim WN_{\mathrm{I}}^{\mathrm{ISH}} \\
                      \Theta(\frac{m}{n}\ell r_{\mathrm{cell}}^{-\alpha})&\textbf{iff }r_{\mathrm{cell}}^{-\alpha}\ll WN_{\mathrm{I}}^{\mathrm{ISH}}.\\
                     \end{cases}
\end{equation}

Finally, we show that the threshold of \eqref{eq:mutinfISHregimes} is equivalent to $\psi>(\beta-\nu)\frac{\alpha}{2}$, producing \eqref{eq:ISHDLR}. The ISH uplink analysis is identical except the transmitted power budget in uplink scales with $n/m$ per cell, yielding \eqref{eq:ISHULR}. 

We examine necessary and sufficient conditions for $r_{\mathrm{cell}}^{-\alpha}\ll WN_{\mathrm{I}}^{\mathrm{ISH}}$ separately. For a necessary condition for $W N_{\mathrm{I}}\ll r_{\mathrm{cell}}^{-\alpha}$, we upper bound the PSD by
\begin{equation}
 \label{eq:NIDL}
 N_{\mathrm{I}}^{\mathrm{ISH}}\leq\Theta\left(n^{\left((\beta-\nu)\frac{\alpha}{2}-\psi\right)^+}\right)
\end{equation}
so, if $\psi\leq(\beta-\nu)\frac{\alpha}{2}$, then the scaling exponent of $WN_{\mathrm{I}}^{\mathrm{ISH}}$ is always lower than or equal to $r_{\mathrm{cell}}^{-\alpha}\leq d_{t,r}^{-\alpha}$, and the rates of all users \eqref{eq:mutinfISH} are bandwidth-limited. Therefore, $\psi>(\beta-\nu)\frac{\alpha}{2}$ is necessary for the link rates to become power-limited.

To prove \eqref{eq:NIDL} we upper bound the distance sum in \eqref{eq:NIx}. In the ISH protocol we can upper bound the distance $d_{t',r}$ by the distance between the border of the cell of BS $t'$ and the cell where $r$ belongs. Considering the geometry of the hexagonal layout it can be shown that there $6k$ cells form a ring with index $k$ separated a distance greater or equal than $\frac{3}{2}k-1$ cell radii from the border of the cell of $r$. The network is finite and a maximum $k$ exists, but we can get rid of border effects and further bound the interference by extending the sum all the way to $k\rightarrow\infty$.
\begin{equation}
\label{eq:concentricsum}
\begin{split}
  \sum_{t'\in\mathcal{I}^{\mathrm{ISH}}}d_{t',r}^{-\alpha}
    &\leq r_{\mathrm{cell}}^{-\alpha}\sum_{k=1}^{\infty}(6k)(2k-1)^{-\alpha}\\
    &\leq r_{\mathrm{cell}}^{-\alpha}\sum_{k=1}^{\infty}\frac{12}{(2k-1)^{\alpha-1}}\\
    &= 12r_{\mathrm{cell}}^{-\alpha}\left(\frac{3}{2}\right)^{\alpha-1}\zeta(\alpha-1,\frac{1}{3})
\end{split}
\end{equation}
where the last equality holds for $\alpha>2$ and the result is expressed using the generalized Riemann Zeta function $\displaystyle \zeta(s)=\sum_{x=1}^{\infty}\frac{1}{x^s}$, which is a constant with regard to $n$. Combining with $W$ this shows that interference PSD scales as $n^{(\beta-\nu)\frac{\alpha}{2}-\psi}$, while noise PSD $N_0$ is constant, so \eqref{eq:NIx} for ISH downlink scales as \eqref{eq:NIDL}.

To prove the same condition is also sufficient we assume $\psi>(\beta-\nu)\frac{\alpha}{2}$. This allows to approximate the upper bound \eqref{eq:NIDL} by $\simeq N_0$, and since $N_{\mathrm{I}}^{ISH}\geq N_0$ the bound is tight. Therefore, if $\psi>(\beta-\nu)\frac{\alpha}{2}$ then $R_{t,r}^{\mathrm{ISH}}\simeq \frac{m}{n}\ell W\log\left(1+\frac{P_{BS}d_{t,r}^{-\alpha}}{W N_0} \right)$, and since $\min_{t,r}d_{t,r}^{-\alpha}=\Theta(r_{\mathrm{cell}}^{-\alpha})\leq\Theta(W)$, this leads to $R^{\mathrm{ISH}}(n)=\Theta(\frac{m}{n}\ell r_{\mathrm{cell}}^{-\alpha})$.

\section{Proof of Theorem \ref{the:IMH}}
\label{sec:IMH}

The proof of Theorem \ref{the:IMH} consists of using the arguments in Appendix \ref{sec:ISH} with minor variations. We only discuss the main differences here. Achievable rate scaling in IMH is given by the minimum of two constraints: we denote the BS-node link rates at the center of the cell by $R_{t,r}^{\mathrm{IMH}(1)}$, and node-node link rates in the rest of the multi-hop system by $R_{t,r}^{\mathrm{IMH}(2)}$. Thus
$$R^{\mathrm{IMH}}_{DL}(n)=\min_{t,r}\min_{\{1,2\}}(R_{t,r}^{\mathrm{IMH}(1)},R_{t,r}^{\mathrm{IMH}(2)}).$$

The scaling of $R_{t,r}^{\mathrm{IMH}(1)}$, the BS-node link rate, is scaling using a MU-MIMO similar to ISH, with the exception that no frequency division is required. The BS of each cell transmits $\ell$ signals towards the nearest routing subcells with bandwidth $W$ using separate spatial signatures. There are a total of $n/m$ routes that need to be served by the BS, so the route towards each destination is time multiplexed by a factor $m\ell/n$ at the BS. Thus, rate between a BS and a node in its nearest routing subcell can be expressed as
\begin{equation}
\label{eq:RIMH1}
R_{t,r}^{\mathrm{IMH(1)}}=\frac{m\ell}{n}W\log\left(1+\frac{P_t d_{t,r}^{-\alpha}\ell}{\Ex{}{|I_1|^2}+W N_{\mathrm{I}}^{\mathrm{IMH}}} \right)
\end{equation}
where the power is $P_t=P_\mathrm{BS}/\ell$ in DL and $P_t=P$ in UL.

The link rates in the rest of the node-node hops of each route are single antenna. Here, a factor of$\frac{m\ell}{n}$ is imposed on the rate because the routes inherit it from the BS. 
\begin{equation}
\label{eq:RIMH2}
R_{t,r}^{\mathrm{IMH}(2)}=\frac{m\ell}{n}W\log\left(1+\frac{Pd_{t,r}^{-\alpha}}{W N_{\mathrm{I}}^{\mathrm{IMH}}}\right)
\end{equation}

In order to determine the scaling of each term, 
the same arguments applied to the scaling of \eqref{eq:mutinfISHregimes} can then be applied to \eqref{eq:RIMH1} and \eqref{eq:RIMH2}. To study $N_\mathrm{I}^{\mathrm{IMH}}$, the noise plus out-of-subcell interference PSD for IMH using \eqref{eq:NIx}, we have to take into account that the set of interferers $\mathcal{I}^{\mathrm{IMH}}_{t,r}$ is the set of all transmitters in nearby routing subcells that transmit at the same time as the link $(t,r)$. In downlink the interferers may be either other nodes or BS, and the power of each interferer can be upper bounded by the worst case $P_{t'}\leq \max(P_\mathrm{BS},P)$, whereas in UL only nodes transmit $P_{t'}=P$.
 
In order to derive an expression similar to \eqref{eq:concentricsum} for IMH, $d_{t',r}$ is upper bounded by the distance between the border of the routing subcell that contains $t'$ and the border of the routing subcell that contains $r$. In order to guarantee a half-duplex collision-free routing, IMH uses a time-division where only 1 out of every 7 subcells in the hexagonal layout transmit at the same time. Following as in \eqref{eq:concentricsum} we can show that
 $$
  N_{\mathrm{I}}^{\mathrm{IMH}}=\Theta\left(n^{\left((1-\nu)\frac{\alpha}{2}-\psi\right)^+}\right)
 $$

Examining Table \ref{tab:exponents} both the BS to nodes and the node to node links can be shown to scale as
 \begin{equation}
 \label{eq:IMHDL1}
  R_{t,r}^{\mathrm{IMH}(1)}=\Theta(n^{\beta+\gamma-1+\min(\psi,(1-\nu)\frac{\alpha}{2})})=R_{t,r}^{\mathrm{IMH}(2)}.
 \end{equation}
where we distinguish the same two regimes as in ISH, but this time the threshold is $\psi<(1-\nu)\frac{\alpha}{2}$ for the bandwidth-limited rates, proving Theorem~\ref{the:IMH} for the downlink. 

The proof for uplink follows the same principle but the BS receives multiple transmissions at the same time, increasing the received power, and the point to point links become a bottleneck where each route is served with rate
 \begin{equation}
  R^{\mathrm{IMH}}_{UL}(n)=\Theta(n^{\beta+\gamma-1+\min(\psi,(1-\nu)\frac{\alpha}{2})})
 \end{equation}

\section{Analysis of IRH}
\label{sec:IRH}

We first assume that the IRH always makes use of the RNs, and from the resulting rates derive the scaling exponent thresholds where it is best to fall back to ISH. The IRH protocol implements RN multi-hop in two phases. The access phase is an ISH protocol with $m+k$ APs that consist of both BSs and RNs. We denote this protocol by ISH-R its achievable rate scaling
\begin{equation}
\begin{split}
 R^{\mathrm{ISH-R}}_{\mathrm{DL}}(n)&=\tau\Theta\left(n^{\rho-1+\min\left(\psi,(\rho-\nu)\frac{\alpha}{2}\right)}\right).\\
 \end{split}
\end{equation}

The interconnection phase uses an IMH protocol where the BSs are the infrastructure and the RNs are the user nodes. We denote this by IMH-R with achievable backhauling rate per RN
\begin{equation}
\begin{split}
 R^{\mathrm{IMH-R}}_{\mathrm{DL}}(k)&=(1-\tau)\Theta\left(k^{\beta'+\gamma'-1+\min\left(\psi',(1-\nu')\frac{\alpha}{2}\right)}\right)\\
 &\textnormal{ where }\beta'=\frac{\beta}{\rho},\; \gamma'=\frac{\gamma}{\rho},\; \psi'=\frac{\psi}{\rho}.
 \end{split}
\end{equation}
where we modify the exponents of area, bandwidth and number of antennas to account for the fact that IMH is evaluated with $k=n^\rho$ user nodes. Each RN serves $n/k$ nodes and must divide the backhauling rate it achieves equally among them, giving $n^{\rho-1}R^{\mathrm{IMH-R}}_{\mathrm{DL}}(n^\rho)$ per user node.  

Finally, each node in the IRH network achieves the minimum between the rate of its link to its AP, and the fraction of the AP bachkauling rate that is assigned to the node.
\begin{equation}
\begin{split}
 R^{\mathrm{IRH}}_{\mathrm{DL}}(n)=\Theta\Big(\min
 \Big(
 &(1-\tau)n^{\beta+\gamma-1+\min\left(\psi,(\rho-\nu)\frac{\alpha}{2}\right)},\\
 &\quad\tau n^{\rho-1+\min\left(\psi,(\rho-\nu)\frac{\alpha}{2}\right)}
\Big)
\Big)
 \end{split}
\end{equation}

\remembertext{RC17}{The optimal time-division in a decode-and-forward relay scheme with two links of known rates is well known to be the value that minimizes the total transmission time per bit, as verified for example in \cite{journals/tit/AzarianGS05}
. This is given by 
\begin{equation}
\begin{split}
 \tau^*&=\arg \min_{\tau} (1-\tau)\frac{1}{k/nR^{\mathrm{IMH-R}}_{\mathrm{DL}}(k)}+\tau\frac{1}{R^{\mathrm{IRH}}_{\mathrm{DL}}(n)}\\
 &=\frac{k/nR^{\mathrm{IMH-R}}_{\mathrm{DL}}(k)}{R^{\mathrm{IRH}}_{\mathrm{DL}}(n)+k/nR^{\mathrm{IMH-R}}_{\mathrm{DL}}(k)}.
 \end{split}
\end{equation}
If $R^{ISH-R}_{DL}(n)>\Theta(k/nR^{IMH-R}_{DL}(n))$, $\tau^*$ converges to zero. If $R^{ISH-R}_{DL}(n)<\Theta(k/nR^{IMH-R}_{DL}(n))$, $\tau^*$ converges to 1. If $R^{ISH-R}_{DL}(n)=\Theta(R^{IMH-R}_{DL}(n))$, $\tau^*$ does not affect the rate scaling.} Putting everything together gives the expression in Theorem~\ref{the:IRH} for downlink. By careful comparison of Theorem~\ref{the:IRH} and Theorem~\ref{the:ISH}, we deduce the threshold where it is better to use ISH is $\rho\geq\beta+\gamma+(\beta-\nu)\frac{\alpha}{2}-\psi$.

By similar arguments, the scaling of uplink IRH when the RNs are used can be shown to be
\begin{equation}
\begin{split}
 R^{\mathrm{IRH}}_{\mathrm{UL}}(n)=\Theta\Big(\min
 \Big(
 &(1-\tau)n^{\beta+\gamma-1+\min\left(\psi,(\rho-\nu)\frac{\alpha}{2}\right)},\\
 &\quad \tau n^{\rho-1+\min\left(\psi,(\rho-\nu)\frac{\alpha}{2}+1-\rho\right)}
\Big)
  \Big)\\
 \end{split}
\end{equation}
and by comparison with Theorem~\ref{the:ISH} we can deduce that IRH only benefits from RNs if $\rho$ is high enough such that $\min\left(\psi-(\beta+\gamma-\rho)^+,(\rho-\nu)\frac{\alpha}{2}\right)\geq\min\left(\psi,(\beta-\nu)\frac{\alpha}{2}+1-\beta\right)$.



\vspace{-1.7in}
\begin{IEEEbiography}
[{\includegraphics[width=1in,height=1.25in,clip,keepaspectratio]{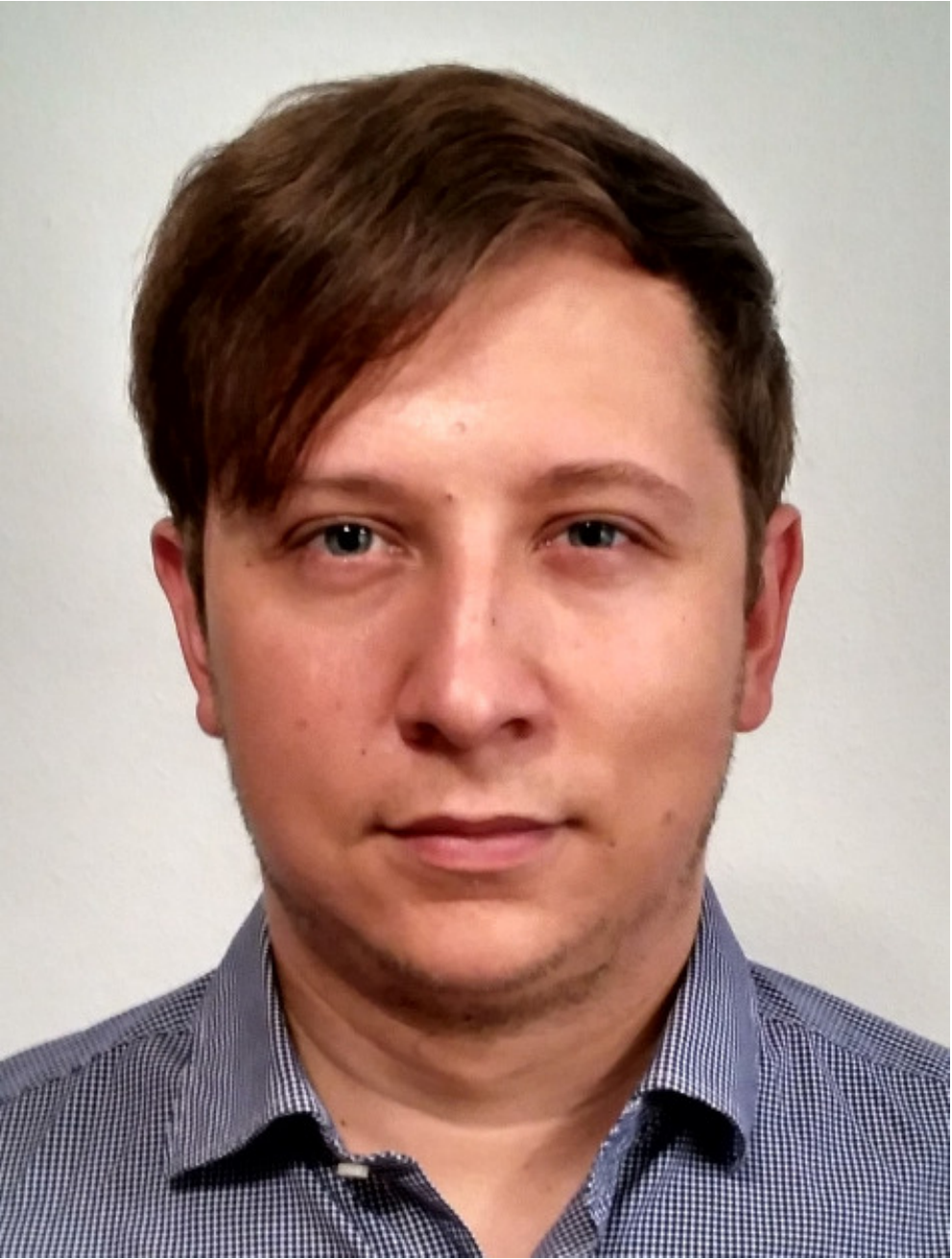}}]
{Felipe G\'omez-Cuba}
received his Ingeniero de Telecomunicaci\'on degree in 2010 (Telecommunication Engineering 5 year program), M.Sc in Signal Processing Applications for Communications in 2012, and a PhD degree in 2015 from the University of Vigo, Spain. He worked as a researcher in the Information Technologies Group (GTI), Telematic Engineering Department (DET), University of Vigo, (2010--2011), the Galician Research and Development center In Advanced Telecommunications (GRADIANT), Vigo, Spain (2011--2013). He also worked as a visiting scholar in the NYUWireless center at NYU Tandon School of Engineering (2013--2014) and completed his PhD under a FPU grant from the Spanish MINECO (2013--2016) at the University of Vigo. He has been awarded a Marie Curie Individual Fellowship - Global Fellowship with the Dipartimento d'Engegneria dell'Informazione, University of Padova, Italy, and the Department of Electrical Engineering, Stanford University, USA (2016-present).
\end{IEEEbiography}

\newpage

\begin{IEEEbiography}
[{\includegraphics[width=1in,height=1.25in,clip,keepaspectratio]{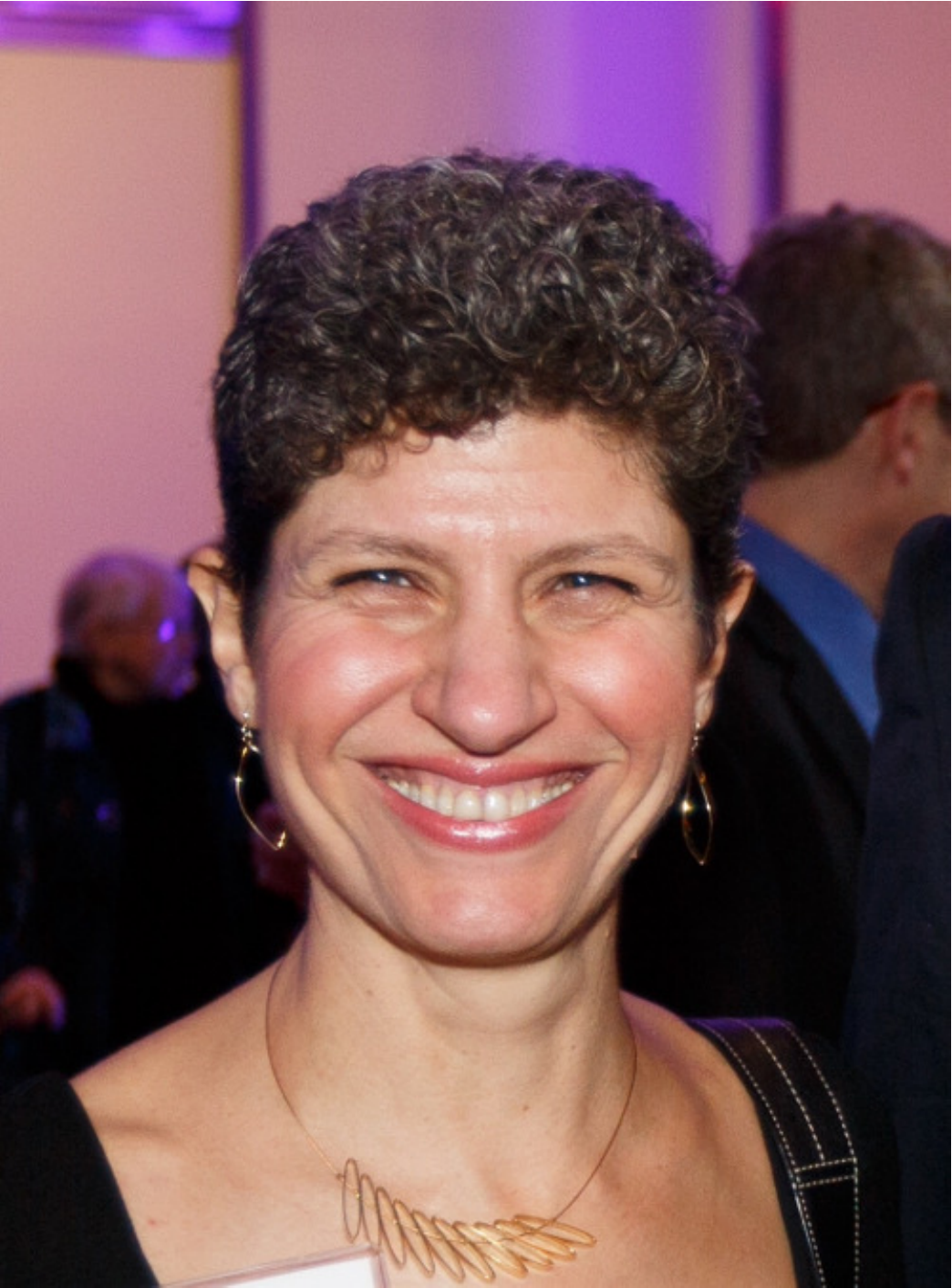}}]
{Elza Erkip}
Elza Erkip (S’93–M’96–SM’05–F’11) received the B.S. degree in Electrical and Electronics Engineering from Middle East Technical University, Ankara, Turkey, and the M.S. and Ph.D. degrees in Electrical Engineering from Stanford University, Stanford, CA, USA. Currently, she is a Professor of Electrical and Computer Engineering with New York University Tandon School of Engineering, Brooklyn, NY, USA. Her research interests are in information theory, communication theory, and wireless communications.

Dr. Erkip is a member of the Science Academy Society of Turkey and is among the 2014 and 2015 Thomson Reuters Highly Cited Researchers. She received the NSF CAREER award in 2001 and the IEEE Communications Society WICE Outstanding Achievement Award in 2016. Her paper awards include the IEEE Communications Society Stephen O. Rice Paper Prize in 2004, and the IEEE Communications Society Award for Advances in Communication in 2013. She has been a member of the Board of Governors of the IEEE Information Theory Society since 2012 where she is currently the First Vice President.   She was a Distinguished Lecturer of the IEEE Information Theory Society from 2013 to 2014. 
\end{IEEEbiography}

\vspace{-3in}

\begin{IEEEbiography}
[{\includegraphics[width=1in,height=1.25in,clip,keepaspectratio]{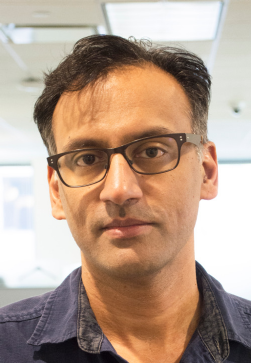}}]
{Sundeep Rangan}

Sundeep Rangan (S’94–M’98–SM’13–F’16) received the B.A.Sc. at the University of Waterloo, Canada and the M.Sc. and Ph.D. at the University of California, Berkeley, all in Electrical Engineering. He has held postdoctoral appointments at the University of Michigan, Ann Arbor and Bell Labs.

In 2000, he co-founded (with four others) Flarion Technologies, a spin off of Bell Labs, that developed Flash OFDM, one of the first cellular OFDM data systems and pre-cursor to 4G systems including LTE and WiMAX.  In 2006, Flarion was acquired by Qualcomm Technologies where Dr. Rangan was a Director of Engineering involved in OFDM infrastructure products. He joined the ECE department at NYU Tandon (formerly NYU Polytechnic) in 2010. He is a Fellow of the IEEE and Director of NYU WIRELESS, an academic-industry research center researching next-generation wireless systems.  His research interests are in wireless communications, signal processing, information theory and control theory.
\end{IEEEbiography}

\vspace{-3in}

\begin{IEEEbiography}
[{\includegraphics[width=1in,height=1.25in,clip,keepaspectratio]{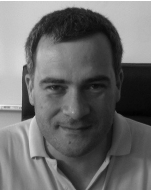}}]
{Francisco Javier Gonz\'alez-Casta\~no}
Francisco J. González-Castaño received the Ingeniero de Telecomunicación degree from University
of Santiago de Compostela, Spain, in 1990 and the Doctor Ingeniero de Telecomunicación degree from
University of Vigo, Spain, in 1998. He is currently a Catedrático de Universidad (Full Professor) with
the Telematics Engineering Department, University of Vigo, Spain, where he leads the Information Technologies Group (http://www-gti.det.uvigo.es).
He has authored over 80 papers in international journals, in the fields of telecommunications and computer science, and has participated in several
relevant national and international projects. He holds two U.S. patents.
\end{IEEEbiography}

\end{document}